\documentclass[a4paper]{article}

\usepackage[top=1.5in, bottom=1.5in, left=1.25in, right=1.25in]{geometry}

\usepackage{makeidx}
\usepackage{latexsym}
\usepackage{amssymb}
\usepackage{amsthm}
\usepackage{amsmath}
\usepackage{textcomp}
\usepackage{multirow}
\usepackage[all]{xy}
\usepackage{algorithmic}
\usepackage{algorithm}
\usepackage{url}
\usepackage{tikz}
\usetikzlibrary{calc}
\usetikzlibrary{external}
\usepackage{authblk}
\usepackage{hyperref}

\newtheorem{theorem}{Theorem}
\newtheorem{lemma}[theorem]{Lemma}
\newtheorem{definition}[theorem]{Definition}

\newtheorem{example}[theorem]{Example}
\newtheorem{proposition}[theorem]{Proposition}
\newtheorem{remark}[theorem]{Remark}

\DeclareMathOperator{\landau}{O}

\DeclareMathOperator{\var}{var}
\DeclareMathOperator{\vd}{vd}
\DeclareMathOperator{\cn}{cn}
\DeclareMathOperator{\jfa}{JFA}
\DeclareMathOperator{\co}{co}
\DeclareMathOperator{\lforth}{l-forth}
\DeclareMathOperator{\rforth}{r-forth}
\DeclareMathOperator{\lback}{l-back}
\DeclareMathOperator{\rback}{r-back}
\DeclareMathOperator{\match}{match}

\newcommand{\cent}{\ensuremath{\text{\textcent}}}

\newcommand{\bbN}{\ensuremath{\mathbb{N}}}

\newcommand{\ta}{\ensuremath{\mathtt{a}}}
\newcommand{\tb}{\ensuremath{\mathtt{b}}}

\newcommand{\varpos}[1]{\ensuremath{\mbox{varpos}_{#1}}}

\newcommand{\PL}{\ensuremath{L_{\Sigma}}}

\begin{document}

\title{A Polynomial Time Match Test for Large Classes of Extended Regular Expressions\thanks{This document is a full version (i.\,e., it contains all proofs) of the conference paper \cite{rei:apo}.}}

\author[1]{Daniel Reidenbach}
\author[1,2]{Markus~L.~Schmid}

\affil[1]{Department of Computer Science, Loughborough University,\\ Loughborough, Leicestershire, LE11 3TU, United Kingdom \texttt{D.Reidenbach@lboro.ac.uk}}
\affil[2]{Fachbereich 4 -- Abteilung Informatikwissenschaften, Universit\"at Trier, 54286 Trier, Germany, \texttt{mschmid@uni-trier.de}}

\maketitle

\begin{abstract}
In the present paper, we study the match test for extended regular expressions. We approach this NP-complete problem by introducing a novel variant of two-way multihead automata, which reveals that the complexity of the match test is determined by a hidden combinatorial property of extended regular expressions, and it shows that a restriction of the corresponding parameter leads to rich classes with a polynomial time match test. For presentational reasons, we use the concept of pattern languages in order to specify extended regular expressions. While this decision, formally, slightly narrows the scope of our results, an extension of our concepts and results to more general notions of extended regular expressions is straightforward.
\end{abstract}

\section{Introduction}\label{sec:intro}
Regular expressions are compact and convenient devices that are widely used to specify regular languages, e.\,g., when searching for a pattern in a string. In order to overcome their limited expressive power while, at the same time, preserving their desirable compactness, their definition has undergone various modifications and extensions in the past decades. These amendments have led to several competing definitions, which are collectively referred to as \emph{extended regular expressions} (or: \emph{REGEX} for short). Hence, today's text editors and programming languages (such as Java and Perl) use individual notions of (extended) regular expressions, and they all provide so-called \emph{REGEX engines} to conduct a \emph{match test}, i.\,e., to compute the solution to the membership problem for any language given by a REGEX and an arbitrary string. While the introduction of new features of extended regular expressions has frequently not been guided by theoretically sound analyses, recent studies have led to a deeper understanding of their properties (see, e.\,g., C\^ampeanu et al.~\cite{cam:afo}).\par
A common feature of extended regular expressions not to be found in the original definition is the option to postulate that each word covered by a specific REGEX must contain a variable substring at several recurrent positions (so-called \emph{backreferences}). Thus, they can be used to specify a variety of non-regular languages (such as the language of all words $w$ that satisfy $w = xx$ for arbitrary words $x$), and this has severe consequences on the complexity of their basic decision problems. In particular, their vital membership problem (i.\,e., in other words, the match test) is NP-complete (see Aho~\cite{aho:alg}). Although this matter is hardly discussed by the literature on the application of extended regular expressions (see, e.\,g., Friedl \cite{fri:mas}), many implementations of REGEX engines impose restrictions on the backreferences -- e.\,g., by limiting their number to $9$ -- in order to manage the trade-off between expressive power and time complexity. Recent developments of REGEX engines that are particularly tailored to efficiency even completly abandon the support of backreferences (see, e.\,g., Google's RE2 \cite{goo:re2} and Le Maout \cite{lem:reg}), so that they can make use of the well-developed theory of finite automata as acceptors of regular languages. On the other hand, the original introduction of backreferences has been motivated by practical needs, which implies that such radical solutions cannot be used in various applied settings. Furthermore, as demonstrated by Freydenberger \cite{fre:ext}, regular expressions with backreferences allow the specification of regular languages in a much more compact manner than their counterparts without backreferences. More precisely, the size difference between extended regular expressions and equivalent ``normal'' regular expressions is not bounded by any recursive function. Thus, users of regular expressions will inevitably wish to specify such languages via backreferences, and the match test for expressions with backreferences can be considerably faster than that for expressions that do not make use of this concept.
We therefore consider it a worthwhile task to investigate alternative approaches to the match test of REGEX with backreferences and to establish large classes of extended regular expressions that have a polynomial-time match test. Moreover, in order to support an integration with existing state-of-the-art REGEX engines that do not support backreferences, it is desirable that the corresponding concepts are based on appropriate automata.\par
It is the purpose of this paper to propose and study such an alternative method. In order to keep the technical details reasonably concise we do not directly use a particular REGEX definition, but we consider a well-established type of formal languages that, firstly, is defined in a similar yet simpler manner, secondly, is a proper subclass of the languages generated by extended regular expressions and, thirdly, shows the same properties with regard to the membership problem: the \emph{pattern languages} as introduced by Angluin~\cite{ang:fin2}; our results can then directly be transferred to extended regular expressions. In this context, a \emph{pattern} $\alpha$ is a finite string that consists of \emph{variables} and \emph{terminal symbols} (taken from a fixed alphabet $\Sigma$), and its language is the set of all words that can be derived from $\alpha$ when substituting arbitrary words over $\Sigma$ for the variables. For example, the language $L$ generated by the pattern $\alpha := x_1 \ta x_2 \tb x_1$ (where $x_1, x_2$ are variables and $\ta, \tb$ are terminal symbols) consists of all words with an arbitrary prefix $u$, followed by the letter $\ta$, an arbitrary word $v$, the letter $\tb$ and a suffix that equals the prefix $u$. Thus, $w_1 := \ta \ta \ta \tb \tb \ta \ta$ is contained in $L$, whereas $w_2 := \tb\ta\ta\tb\ta$ is not.\par
In the definition of pattern languages, the option of using several occurrences of a variable exactly corresponds to the backreferences in extended regular expressions, and therefore the membership problem for pattern languages captures the essence of what is computationally complex in the match test for REGEX. Thus, it is not surprising that the membership problem for pattern languages is also known to be NP-complete (see Angluin~\cite{ang:fin2} and Jiang et al.~\cite{jia:pat}). Furthermore, Ibarra et al.~\cite{iba:ano} point out that the membership problem for pattern languages is closely related to the solvability problem for certain Diophantine equations. More precisely, for any word $w$ and for any pattern $\alpha$ with $m$ terminal symbols and $n$ different variables, $w$ can only be contained in the language generated by $\alpha$ if there are numbers $s_i$ (representing the lengths of the substitution words for the variables $x_i$) such that $|w| = m + \sum_{i=1}^n a_i s_i$ (where $a_i$ is the number of occurrences of $x_i$ in $\alpha$ and $|w|$ stands for the \emph{length} of $w$). Thus, the membership test needs to implicitly solve this NP-complete problem, which is called \emph{Money-Changing} or \emph{Coin Problem} and -- due to its fundamentality and its practical relevance, e.\,g., in Operations Research -- has been intensively studied. All these insights into the complexity of the membership problem do not depend on the question of whether the pattern contains any terminal symbols. Therefore, we can safely restrict our considerations to so-called \emph{terminal-free} pattern languages (generated by patterns that consist of variables only); for this case, NP-completeness of the membership problem has indirectly been established by Ehrenfeucht and Rozenberg~\cite{ehr:fin}. This restriction again improves the accessibility of our technical concepts, without causing a loss of generality.\par
As stated above, these results on the complexity of the problem (and the fact that probabilistic solutions might often be deemed inappropriate for it) motivate the search for large subclasses with efficiently solvable membership problem and for suitable concepts realising the respective algorithms. Rather few such classes are known to date. They either restrict the number of \emph{different} variables in the patterns to a fixed number $k$ (see Angluin~\cite{ang:fin2}, Ibarra et al.~\cite{iba:ano}), which is an obvious option and leads to a time complexity of $\landau(n^k)$, or they restrict the number of \emph{occurrences} of each variable to $1$ (see Shinohara~\cite{shi:pol1}), which turns the resulting pattern languages into regular languages.
\par
In the present paper, motivated by Shinohara's~\cite{shi:pol2} \emph{non-cross} pattern languages, we introduce major classes of pattern languages (and, hence, of extended regular expressions) with a polynomial-time membership problem that do not show any of the above limitations. Thus, the corresponding patterns can have any number of variables with any number of occurrences; instead, we consider a rather subtle parameter, namely the \emph{distance} several occurrences of any variable $x$ may have in a pattern (i.\,e., the maximum number of different variables separating any two consecutive occurrences of $x$). We call this parameter the \emph{variable distance} $\vd$ of a pattern, and we demonstrate that, for the class of all patterns with $\vd \leq k$, the membership problem is solvable in time $\landau(n^{k + 4})$. Referring to the proximity between the subject of our work and the solvability problem of Diophantine equations (which does not depend on the order of variables in the patterns, but merely on their numbers of occurrences), we consider this insight quite remarkable, and it is only possible since the Money Changing Problem is \emph{weakly} NP-complete, i.\,e., it is only NP-complete since its input merely consists of numbers in binary representation, which means that the input length for the Money Changing Problem is exponentially smaller than for the membership problem for pattern languages, where we have to regard the lengths of the input strings as input length of the problem. We also wish to point out that, in terms of our concept, Shinohara's non-cross patterns correspond to those patterns with $\vd = 0$.\par
We prove our main result by introducing the concept of a \emph{Janus automaton}, which is a variant of a two-way two-head automaton (see Ibarra~\cite{iba:ont}), amended by the addition of a number of counters. Janus automata are algorithmic devices that are tailored to performing the match test for pattern languages, and we present a systematic way of constructing them. While an intuitive use of a Janus automaton assigns a distinct counter to each variable in the corresponding pattern $\alpha$, we show that in our advanced construction the number of different counters can be limited by the variable distance of $\alpha$. Since the number of counters is the main element determining the complexity of a Janus automaton, this yields our main result. An additional effect of the strictness of our approach is that we can easily discuss its quality in a formal manner, and we can show that, based on a natural assumption on how Janus automata operate, our method leads to an automaton with the smallest possible number of counters.\par
This paper is organised as follows. In Section~\ref{sec:definitions} the basic definitions and the concept of pattern languages are introduced. The purpose of Section~\ref{sec:janus} is to introduce our new model, the Janus automaton. In Section~\ref{sec:pattern} we show how we can effectively construct Janus automata to recognise pattern languages. Then, in Section~\ref{sec:varDist}, the above mentioned concept of the variable distance is introduced. Also in this section, we shall use this notion to present and prove our main result. Finally, we summarise this paper in Section~\ref{sec:conclusion}, and we give an overview of related and further research ideas.

\section{Basic Definitions}\label{sec:definitions}

Let $\bbN := \{0,1, 2, 3, \ldots \}$. For an arbitrary alphabet $A$, a \emph{string} (\emph{over $A$}) is a finite sequence of symbols from $A$, and $\varepsilon$ stands for the \emph{empty string}. The symbol $A^+$ denotes the set of all nonempty strings over $A$, and $A^*:=A^+ \cup \{ \varepsilon \}$. For the \emph{concatenation} of two strings $w_1, w_2$ we write $w_1 \cdot w_2$ or simply $w_1 w_2$. We say that a string $v \in A^*$ is a \emph{factor} of a string $w \in A^*$ if there are $u_1, u_2 \in A^*$ such that $w = u_1 \cdot v \cdot u_2$. The notation $|K|$ stands for the size of a set $K$ or the length of a string $K$; the term $|w|_a$ refers to the number of occurrences of the symbol $a$ in the string $w$. \par
For any alphabets $A, B$, a \emph{morphism} is a function $h: A^*\rightarrow B^*$ that satisfies $h(vw)=h(v)h(w)$ for all $v,w\in A^*$. Let $\Sigma$ be a (finite) alphabet of so-called \emph{terminal symbols} and $X$ an infinite set of \emph{variables} with $\Sigma \cap X = \emptyset$. We normally assume $X := \{ x_1, x_2, x_3, \ldots \}$. A \emph{pattern} is a nonempty string over $\Sigma \cup X$, a \emph{terminal-free pattern} is a nonempty string over $X$ and a \emph{word} is a string over $\Sigma$. For any pattern $\alpha$, we refer to the set of variables in $\alpha$ as $\var(\alpha)$. We shall often consider a terminal-free pattern in its variable factorisation, i.\,e.\ $\alpha = y_1 \cdot y_2 \cdot \ldots \cdot y_n$ with $y_i \in \{x_1, x_2, \ldots, x_{m}\}$, $1 \leq i \leq n$ and $m = |\var(\alpha)|$. \par
A morphism $\sigma:\left(\Sigma\cup X\right)^*\rightarrow\Sigma^*$ is called a \emph{substitution} if $\sigma(a)=a$ for every $a\in\Sigma$. We define the \emph{pattern language} of a terminal-free pattern $\alpha$ by $\PL(\alpha) := \{ \sigma(\alpha) \mid \sigma: X^* \to \Sigma^* \mbox{ is a substitution} \}$. Note, that these languages, technically, are terminal-free E-pattern languages (see Jiang et al.~\cite{jia:pat}). Since in our work the dependency on the alphabet $\Sigma$ is negligible, we mostly denote pattern languages by $L(\alpha)$. Furthermore, we ignore all patterns $\alpha$ satisfying, for a variable $x$, $|\alpha|_x = 1$, as then $\PL(\alpha) = \Sigma^*$.\par
The problem to decide for a given pattern $\alpha$ and a given word $w \in \Sigma^*$ whether $w \in L(\alpha)$ is called the \emph{membership problem}.\par
Finally, we assume the reader to be familiar with the basic concepts of automata theory and refer to Hopcroft et al. \cite{hop:int2} for terms not defined explicitly.

\section{Janus Automata}\label{sec:janus}

In order to prove the main results of this paper, we introduce a novel type of automata, the so-called Janus automata, that are tailored to solving the membership problem for pattern languages. We shall first explain this model in an informal way and then give a formal definition.\par
A Janus automaton is a two-way automaton with two input heads. In addition to that, a Janus automaton has a constant number of restricted counters. In every step of the computation the Janus automaton provides a distinct counter bound for every counter. The counter values can only be incremented or left unchanged, and they count strictly modulo their counter bound, i.\,e., once a counter value has reached its counter bound, a further increment forces the counter to start at counter value $0$ again. Depending on the current state, the currently scanned input symbols and on whether the counters have reached their bounds, the transition function determines the next state, the input head movements and the counter instructions, and this is done completely deterministically. In addition to the counter instructions of incrementing and leaving the counter unchanged, it is also possible to reset a counter. In this case, the counter value is set to $0$ and a new counter bound is nondeterministically guessed. Furthermore, we require the first input head to be always positioned to the left of the second input head, so there are a well-defined left and right head. This explains why we call this model a \emph{Janus} automaton.\par
Any string $\cent w \$$, where $w \in \Sigma^*$ and the symbols $\cent, \$$ (referred to as left and right endmarker, respectively) are not in $\Sigma$, is an input. Initially, the input tape stores some input $w$, the automaton is in the initial state, all counter bounds and counter values are $0$ and both input heads scan $\cent$.
The word $w$ is accepted by an automaton if and only if it is possible to reach an accepting state by successively applying the transition function. \par
Janus automata are nondeterministic, but their nondeterminism differs from that of common nondeterministic finite automata. The only nondeterministic step a Janus automaton is able to perform consists in guessing a new counter bound for some counter. Once a new counter bound is guessed, the previous one is lost. Apart from that, each transition, i.\,e., entering a new state, moving the input heads and giving instructions to the counters, is defined completely deterministically.\par
The vital point of a computation of a Janus automaton with $k$ counters is then that the automaton is only able to store exactly $k$ (a constant number, not depending on the input word) different numbers at a time (the counter bounds). We shall see that this number of counters is the crucial number for the complexity of the \emph{acceptance problem}, i.\,e., to decide, for a given word $w$, whether $w$ is in the language accepted by the automaton. \par
We are now ready to present a formal definition of Janus automata:

\begin{definition}
A \emph{Janus automaton with $k$ counters} (denoted by \emph{$\jfa(k)$} in the sequel) is a device $M := (k, Q, \Sigma, \delta, q_0, F)$, where $k \geq 0$ is the number of \emph{counters}, $Q$ is a finite nonempty set of \emph{states}, $\Sigma$ is a finite nonempty alphabet of \emph{input symbols}, $q_0 \in Q$ is the \emph{initial state}, $F \subseteq Q$ is the set of \emph{accepting states} and $\delta$ is a mapping $Q \times \Sigma^{2} \times \{\mathtt{t_{=}}, \mathtt{t_{<}}\}^{k} \rightarrow Q \times \{\mathtt{-1}, \mathtt{0}, \mathtt{1}\}^2 \times \{\mathtt{0}, \mathtt{1}, \mathtt{r}\}^{k}$. The mapping $\delta$ is called the \emph{transition function}.\par
An \emph{input} to $M$ is any string of the form $\cent w \$$, where $w \in \Sigma^*$ and the symbols $\cent, \$$ (referred to as \emph{left} and \emph{right endmarker}, respectively) are not in $\Sigma$. Let $\delta(p, a_1, a_2, s_1, \hdots, s_k) = (q, r_1, r_2, d_1, \hdots, d_k)$. For each $i \in \{1, 2\}$, we call the element $a_i$ the \emph{input symbol scanned by head $i$} and $r_i$ the \emph{instruction for head $i$}. For each $j \in \{1, 2, \ldots, k\}$, the element $s_j \in \{\mathtt{t_{=}}, \mathtt{t_{<}}\}$ is the \emph{counter message of counter $j$}, and $d_j$ is called the \emph{counter instruction for counter $j$}.
\end{definition}
The transition function $\delta$ of a $\jfa(k)$ determines whether the input heads are moved to the left ($r_i = \mathtt{-1}$), to the right ($r_i = \mathtt{1}$) or left unchanged ($r_i = \mathtt{0}$), and whether the counters are incremented ($d_j = \mathtt{1}$), left unchanged ($d_j = \mathtt{0}$) or reset ($d_j = \mathtt{r}$). Next, in order to define the language accepted by a Janus automaton, we need to introduce the concept of a $\jfa(k)$ computation.
\begin{definition} \label{computationDefinition}
Let $M := (k, Q, \Sigma, \delta, q_0, F)$ be a $\jfa(k)$ and $w := b_1 \cdot b_2 \cdot \ldots \cdot b_{n}$, $b_i \in \Sigma$, $1 \leq i \leq n$. A \emph{configuration of $M$ (on input $\cent w\$$)} is an element of the set
\begin{align*}
\widehat{C}_M := \{(q, h_1, h_2, (c_1, C_1), \hdots, (c_k, C_k))~|~&q \in Q, 0 \leq h_1 \leq h_2 \leq n + 1,\\
&0 \leq c_i \leq C_i \leq n, 1 \leq i \leq k\}\enspace.
\end{align*}
The pair $(c_i, C_i)$, $1 \leq i \leq k$, describes the current configuration of the i$^{th}$ counter, where $c_i$ is the \emph{counter value} and $C_i$ the \emph{counter bound}. The element $h_i$, $i \in \{1, 2\}$, is called the \emph{head position of head $i$}.\par
An \emph{atomic move of $M$ (on input $\cent w\$$)} is denoted by the relation $\vdash_{M, w}$ over the set of configurations. Let $\delta(p, a_1, a_2, s_1, \hdots, s_k) = (q, r_1, r_2, d_1, \hdots, d_k)$. Then, for all $c_i, C_i$, $1 \leq i \leq k$, where $c_i < C_i$ if $s_i = \mathtt{t_{<}}$ and $c_i = C_i$ if $s_i = \mathtt{t_{=}}$, and for all $h_1$, $h_2$, $0 \leq h_1 \leq h_2 \leq n + 1$, with $b_{h_i} = a_i$, $i \in \{1, 2\}$, we define $(p, h_1, h_2, (c_1, C_1), \hdots, (c_k, C_k))$ $\vdash_{M, w}$ $(q, h'_1, h'_2, (c'_1, C'_1), \hdots, (c'_k, C'_k))$. Here, the elements $h'_i$, $i \in \{1, 2\}$, and $c'_j, C'_j$, $1 \leq j \leq k$, are defined as follows:
\begin{align*}
h'_i:=&\begin{cases}
h_i + r_i& \mbox{if~$0 \leq h_1 + r_1 \leq h_2 + r_2 \leq n + 1$}\enspace,\\
h_i& \mbox{else}\enspace.
\end{cases}
\end{align*}
For each $j \in \{1, \hdots, k\}$, if $d_j = \mathtt{r}$, then $c'_j := 0$ and, for some $m \in \{0, 1, \ldots, n\}$, $C'_j := m$. If, on the other hand, $d_j \neq \mathtt{r}$, then $C'_j := C_j$ and 
\begin{equation*}
c'_j := c_j + d_j \mod (C_j + 1)\,.
\end{equation*}
To describe a \emph{sequence of (atomic) moves of $M$ (on input $w$)} we use the reflexive and transitive closure of the relation $\vdash_{M, w}$, denoted by $\vdash^*_{M, w}$.
$M$ accepts the word $w$ if and only if $\widehat{c}_0 \vdash^*_{M, w} \widehat{c}_f$, where $\widehat{c}_0 := (q_0, 0, 0, (0, 0), \hdots$, $(0, 0))$ is the \emph{initial configuration}, and $\widehat{c}_f := (q_f, h_1, h_2, (c_1, C_1), \hdots$ $(c_k, C_k))$ is a \emph{final configuration}, for some $q_f \in F$, $0 \leq h_1 \leq h_2 \leq n + 1$ and $0 \leq c_i \leq C_i \leq n$, $1 \leq j \leq k$. For any Janus automaton $M$, let $L(M)$ denote the set of words accepted by $M$.
\end{definition}
By definition, the movements of the input heads are bounded by the left and right endmarkers and the first head is always positioned to the left of the second head. The two counter messages denote whether the counter value is still less than the counter bound ($\mathtt{t_{<}}$) or equals the counter bound ($\mathtt{t_{=}}$). If $\mathtt{r}$ is used in order to reset a counter, a new counter bound is nondeterministically guessed and the counter value is set to $0$.\par
In our applications of this automata model, we use the counters in a particular but natural way. Let us assume that $n$ is the counter bound of a certain counter with counter value $0$. We can define the transition function in a way such that an input head is successively moved to the right and, in every step, the counter is incremented. As soon as the counter reaches its counter bound (i.\,e., its counter message changes from $\mathtt{t_{<}}$ to $\mathtt{t_{=}}$) we stop that procedure and can be sure that the input head has been moved exactly $n$ steps. In this way an automaton can scan whole factors of the input, induced by counter bounds. Furthermore, as we have two input heads, we can use the counter with bound $n$ to move them simultaneously to the right, checking symbol by symbol whether two factors of equal length are the same. It is also worth mentioning that we can use counters in the same way to move input heads from right to left instead of from left to right.\par
This way of using counters shall be made clear by sketching how a Janus automaton $M$ could be defined that recognises the language 
\begin{equation*}
L := \{u \cdot a \cdot v \cdot b \cdot v \cdot u \mid u, v \in \{a, b\}^*\}\,.
\end{equation*}
The Janus automaton $M$ uses two counters and applies the following strategy to check whether an input word $w$ is in $L$. First, we reset both counters and therefore guess two new counter bounds $C_1$ and $C_2$. Then we check if $w = u \cdot a \cdot v \cdot b \cdot v \cdot u$ with $|u| = C_1$ and $|v| = C_2$. This is done by using the first counter to move the right head from position $1$ (the symbol next to the left endmarker) to the right until it reaches position $C_1 + 1$. Then it is checked whether $a$ occurs at this position. After that, by using the second counter, the right head is moved further to the right to position $C_1 + C_2 + 2$, where $M$ checks for the occurrence of the symbol $b$. Next, again by using the second counter, the right head is moved another $C_2 + 1$ steps to the right in order to place it exactly where we expect the second occurrence of factor $u$ to begin. Now, both input heads are moved simultaneously to the right for $C_1$ steps, checking in each step whether they scan the same symbol and whether after these $C_1$ steps the right head scans exactly the right endmarker. If this is successful, we know that $w$ is of form $u \cdot a \cdot v \cdot b \cdot v' \cdot u$, with $|u| = C_1$ and $|v| = |v'| = C_2$. Hence, it only remains to check whether or not $v = v'$. This can be done by positioning both heads at the first positions of the factors $v$ and $v'$, i.\,e., moving the left head one step to the right and the right head $C_1 + C_2$ steps back to the left. In order to perform this, as well as the final matching of the factors $v$ and $v'$, $M$ can apply its counters in the same way as before. If this whole procedure is successful, $M$ shall enter an accepting state, and reject its input otherwise.\par
It is obvious that $w \in L$ if and only if there is a possibility to guess counter bounds such that $M$ accepts $w$; thus, $L(M) = L$.

\section{Janus Automata for Pattern Languages}\label{sec:pattern}

In this section, we demonstrate how Janus automata can be used for recognising pattern languages. More precisely, for an arbitrary terminal-free pattern $\alpha$, we construct a $\jfa(k)$ $M$ satisfying $L(M) = L(\alpha)$. Before we move on to a formal analysis of this task, we discuss the problem of deciding whether $w \in L(\alpha)$ for given $\alpha$ and $w$, i.\,e., the membership problem, in an informal way.\par
Let $\alpha = y_1 \cdot y_2 \cdot \ldots \cdot y_n$ be a terminal-free pattern with $m := |\var(\alpha)|$, and let $w \in \Sigma^*$ be a word. The word $w$ is an element of $L(\alpha)$ if and only if there exists a factorisation $w = u_1 \cdot u_2 \cdot \ldots \cdot u_n$ such that $u_j = u_{j'}$ for all $j, j'$, $1 \leq j < j' \leq |\alpha|$, with $y_j = y_{j'}$. We call such a factorisation $w = u_1 \cdot u_2 \cdot \ldots \cdot u_n$ a \emph{characteristic factorisation for $w \in L(\alpha)$} (or simply \emph{characteristic factorisation} if $w$ and $\alpha$ are obvious from the context). Thus, a way to solve the membership problem is to initially guess $m$ numbers $l_1, l_2, \ldots, l_{m}$, then, if possible, to factorise $w = u_1 \cdot \ldots \cdot u_n$ such that $|u_j| = l_i$ for all $j$ with $y_j = x_i$ and, finally, to check whether this is a characteristic factorisation for $w \in L(\alpha)$. A $\jfa(m)$ can perform this task by initially guessing $m$ counter bounds, which can be interpreted as the lengths of the factors. The two input heads can be used to check if this factorisation has the above described properties.
However, the number of counters that are then required directly depends on the number of variables, and the question arises if this is always necessary.\par
In the next definitions, we shall establish the concepts that formalise and generalise the way of checking whether or not a factorisation is a characteristic one.
\begin{definition}\label{matchingOrderDef}
Let $\alpha := y_1 \cdot y_2 \cdot \ldots \cdot y_n$ be a terminal-free pattern, and, for each $x_i \in \var(\alpha)$, let $n_i := |\alpha|_{x_i}$. The set \emph{$\varpos{i}(\alpha)$} is the set of all positions $j$ satisfying $y_j = x_i$. The sequence $((l_1, r_1), (l_2, r_2), \ldots, (l_{n_i - 1}, r_{n_i - 1}))$ with $(l_j, r_j) \in \varpos{i}(\alpha)^2$ and $l_j < r_j$, $1 \leq j \leq n_i - 1$, is a \emph{matching order for $x_i$ in $\alpha$} if and only if the graph $(\varpos{i}(\alpha), \{(l_1, r_1), (l_2, r_2), \ldots, (l_{n_i - 1}, r_{n_i - 1})\})$ is connected. 
\end{definition}

We consider an example in order to illustrate Definition~\ref{matchingOrderDef}. If, for some pattern $\alpha$ and some $x_i \in \var(\alpha)$, $\varpos{i}(\alpha) := \{1, 3, 5, 9, 14\}$, then the sequences $((5, 1)$, $(14, 3)$, $(1, 3)$, $(9, 3))$, $((1, 3)$, $(3, 5)$, $(5, 9)$, $(9, 14))$ and $((5, 1)$, $(5, 3)$, $(5, 9)$, $(5, 14))$ are some of the possible matching orders for $x_i$ in $\alpha$, whereas the sequences $((1, 3), (9, 1), (3, 9), (5, 14))$ and $((1, 3), (3, 5), (5, 9), (9, 1))$ do not satisfy the conditions to be matching orders for $x_i$ in $\alpha$.\par
To obtain a matching order for a whole pattern $\alpha$ we simply combine matching orders for all $x \in \var(\alpha)$:
\begin{definition}\label{completeMatchingOrderDefinition}
Let $\alpha$ be a terminal-free pattern with $m := |\var(\alpha)|$ and, for all $i$ with $1 \leq i \leq m$, $n_i := |\alpha|_{x_i}$ and let $(m_{i, 1},m_{i, 2}, \ldots, m_{i, n_i - 1})$ be a matching order for $x_i$ in $\alpha$. The tuple $(m_1, m_2, \ldots, m_k)$ is a \emph{complete matching order for $\alpha$} if and only if $k = \sum_{i = 1}^{m} (n_i - 1)$ and, for all $i, j_i$, $1 \leq i \leq m$, $1 \leq j_i \leq n_i - 1$, there is a $j'$, $1 \leq j' \leq k$, with $m_{j'} = m_{i, j_i}$. The elements $m_j \in \varpos{i}(\alpha)^2$ of a matching order $(m_1, m_2, \ldots, m_k)$ are called \emph{matching positions}.
\end{definition}

We introduce an example pattern
\begin{equation*}
\beta := x_1 \cdot x_2 \cdot x_1 \cdot x_2 \cdot x_3 \cdot x_2 \cdot x_3\,,
\end{equation*}
which we shall use throughout the whole paper in order to illustrate the main definitions. Regarding Definition~\ref{completeMatchingOrderDefinition}, we observe that all possible sequences of the matching positions in $\{(1,3)$, $(2,4)$, $(4,6)$, $(5,7)\}$ are some of the possible complete matching orders for $\beta$. As pointed out by the following lemma, the concept of a complete matching order can be used to check whether a factorisation is a characteristic one.
\begin{lemma}\label{matchingOrderLemma}
Let $\alpha = y_1 \cdot y_2 \cdot \ldots \cdot y_n$ be a terminal-free pattern and let $((l_1, r_1)$, $(l_2, r_2)$, $\ldots, (l_{k}, r_{k}))$ be a complete matching order for $\alpha$. Let $w$ be an arbitrary word in some factorisation $w = u_1 \cdot u_2 \cdot \ldots \cdot u_n$. If $u_{l_j} = u_{r_j}$ for every $j$, $1 \leq j \leq k$, then $w = u_1 \cdot u_2 \cdot \ldots \cdot u_n$ is a characteristic factorisation.
\end{lemma}

\begin{proof}
Let $x_i \in \var(\alpha)$ be arbitrarily chosen and let the sequence $((l'_1, r'_1), (l'_2, r'_2),$ $\ldots, (l'_{k\rq{}}, r'_{k\rq{}}))$ be an arbitrary matching order for $x_i$ in $\alpha$. Assume that $u_{l'_j} = u_{r'_j}$ for all $j$, $1 \leq j \leq k\rq{}$. As $(\varpos{i}(\alpha), \{(l'_1, r'_1), (l'_2, r'_2), \ldots, (l'_{k\rq{}}, r'_{k\rq{}})\})$ is a connected graph and as the equality of words is clearly a transitive relation, we can conclude that $u_j = u_{j'}$ for all $j, j'$, $1 \leq j < j' \leq |\alpha|$, with $y_j = y_{j'} = x_i$. Applying this argumentation to all variables in $\alpha$ implies the statement of Lemma~\ref{matchingOrderLemma}. 
\end{proof}

With respect to the complete matching order $((4,6), (1,3), (2,4), (5,7))$ for the example pattern $\beta$, we apply Lemma~\ref{matchingOrderLemma} in the following way. If $w$ can be factorised into $w = u_1 \cdot u_2 \cdot \ldots \cdot u_7$ such that $u_4 = u_6$, $u_1 = u_3$, $u_2 = u_4$ and $u_5 = u_7$, then $w \in L(\beta)$.\par
Let $(l_1, r_1)$ and $(l_2, r_2)$ be two consecutive matching positions of a complete matching order. It is possible to perform the comparison of factors $u_{l_1}$ and $u_{r_1}$ by positioning the left head on the first symbol of $u_{l_1}$, the right head on the first symbol of $u_{r_1}$ and then moving them simultaneously over these factors from left to right, checking symbol by symbol if these factors are identical (cf. the example Janus automaton in Section~\ref{sec:janus}). After that, the left head, located at the first symbol of factor $u_{l_1 + 1}$, has to be moved to the first symbol of factor $u_{l_2}$. If $l_1 < l_2$, then it is sufficient to move it over all the factors $u_{l_1 + 1}, u_{l_1 + 2}, \ldots, u_{l_2 - 1}$. If, on the other hand, $l_2 < l_1$, then the left head has to be moved to the left, and, thus, over the factors $u_{l_1}$ and $u_{l_2}$ as well. Furthermore, as we want to apply these ideas to Janus automata, the heads must be moved in a way that the left head is always located to the left of the right head. The following definition shall formalise these ideas. \par

\begin{definition}\label{janusOperatingModeDefinition}
In the following definition, let $\lambda$ and $\rho$ be constant markers. For all $j, j' \in \mathbb{N}$ with $j < j'$, we define a mapping $g$ by $g(j, j') := (j + 1, j + 2, \ldots, j' - 1)$ and $g(j', j) := (j', j' - 1, \ldots, j)$.\par
Let $((l_1, r_1), (l_2, r_2), \ldots, (l_k, r_k))$ be a complete matching order for a terminal-free pattern $\alpha$ and let $l_0 := r_0 := 0$. For every matching position $(l_i, r_i)$, $1 \leq i \leq k$, we define a sequence $D^{\lambda}_i$ and a sequence $D^{\rho}_i$ by 
\begin{align*}
D^{\lambda}_i &:= ((p_1, \lambda), (p_2, \lambda), \ldots, (p_{k_1}, \lambda)) \text{ and}\\
D^{\rho}_i &:= ((p'_1, \rho), (p'_2, \rho), \ldots, (p'_{k_2}, \rho))\,,
\end{align*}
where $(p_1, p_2, \ldots, p_{k_1}) := g(l_{i - 1}, l_{i})$, $(p'_1, p'_2, \ldots, p'_{k_2}) := g(r_{i - 1}, r_{i})$.\par
Now let $D'_i := ((s_1, \mu_1), (s_2, \mu_2), \ldots, (s_{k_1 + k_2}, \mu_{k_1 + k_2}))$ be a tuple satisfying the following two conditions. Firstly, it contains exactly the elements of $D^{\lambda}_i$ and $D^{\rho}_i$ such that the relative orders of the elements in $D^{\lambda}_i$ and $D^{\rho}_i$ are preserved. Secondly, for every $j$, $1 \leq j \leq k_1 + k_2$, $s_{j_l} \leq s_{j_r}$ needs to be satisfied, with $j_l = \max(\{j' \mid 1 \leq j' \leq j, \mu_{j'} = \lambda\} \cup \{j'_l\})$ and $j_r = \max(\{j' \mid 1 \leq j' \leq j, \mu_{j'} = \rho\} \cup \{j'_r\})$, where $(s_{j'_l}, \mu_{j'_l})$ and $(s_{j'_r}, \mu_{j'_r})$ are the leftmost elements of $D'_i$ with $\mu_{j'_l} = \lambda$ and $\mu_{j'_r} = \rho$, respectively.\par
Now we append the two elements $(r_{i}, \rho)$, $(l_{i}, \lambda)$ in exactly this order to the end of $D'_i$ and obtain $D_i$. Finally, the tuple $(D_1, D_2, \ldots, D_k)$ is called a \emph{Janus operating mode for $\alpha$ (derived from the complete matching order $((l_1, r_1), (l_2, r_2), \ldots, (l_k, r_k))$)}.
\end{definition}

We once again consider the example $\beta = x_1 \cdot x_2 \cdot x_1 \cdot x_2 \cdot x_3 \cdot x_2 \cdot x_3$. According to Definition~\ref{janusOperatingModeDefinition} we consider the tuples $D^{\lambda}_i$ and $D^{\rho}_i$ with respect to the complete matching order $((4,6), (1,3), (2,4), (5,7))$ for $\beta$. We omit the markers $\lambda$ and $\rho$ for a better presentation. 
The tuples $D^{\lambda}_i$ and $D^{\rho}_i$, $1 \leq i \leq 4$, are given by 
\begin{align*}
D^{\lambda}_1 &= (1, 2, 3)\,, &D^{\rho}_1 &= (1, 2, \ldots, 5)\,,\\
D^{\lambda}_2 &= (4, 3, 2, 1)\,, &D^{\rho}_2 &= (6, 5, 4, 3)\,,\\
D^{\lambda}_3 &= ()\,, &D^{\rho}_3 &= ()\,,\\
D^{\lambda}_4 &= (3, 4)\,, &D^{\rho}_4 &= (5, 6)\,.
\end{align*}
Therefore, $\Delta_{\beta} := (D_1, D_2, D_3, D_4)$ is a possible Janus operating mode for $\beta$ derived from $((4,6)$, $(1,3)$, $(2,4)$, $(5,7))$, where 
\begin{align*}
D_1 &= ((1, \rho), (1, \lambda), (2, \rho), (2, \lambda), (3, \rho), (3, \lambda), (4, \rho), (5, \rho), (6, \rho), (4, \lambda)),\\
D_2 &= ((4, \lambda), (3, \lambda), \ldots, (1, \lambda), (6, \rho), (5, \rho), \ldots, (3, \rho), (3, \rho), (1, \lambda)),\\
D_3 &= ((4, \rho), (2, \lambda)),\\
D_4 &= ((3, \lambda), (5, \rho), (4, \lambda), (6, \rho), (7, \rho), (5, \lambda)).
\end{align*}
Intuitively, we interpreted a complete matching order as a list of instructions specifying how to check whether a factorisation is a characteristic one. Similarly, a Janus operating mode derived from a complete matching order can be seen as an extension of this complete matching order that also contains information of how two input heads have to be moved from one matching position to the next one. Hence, there is an immediate connection between Janus operating modes and Janus automata for terminal-free pattern languages, and we shall see that it is possible to transform a Janus operating mode for any pattern directly into a Janus automaton recognising the corresponding pattern language. As we are particularly interested in the number of counters a Janus automaton needs, we introduce an instrument to determine the quality of Janus operating modes with respect to the number of counters that are required to actually construct a Janus automaton.

\begin{definition}\label{counterNumber}
Let $\Delta_{\alpha} := (D_1, D_2, \ldots, D_k)$ be a Janus operating mode for a terminal-free pattern $\alpha := y_1 \cdot y_2 \cdot \ldots \cdot y_n$. The \emph{head movement indicator} of $\Delta_{\alpha}$ is the tuple $\overline{\Delta_{\alpha}} = ((d'_1, \mu'_1)$, $(d'_2, \mu'_2)$, $\ldots, (d'_{k'}, \mu'_{k'}))$ with $k' = \sum_{i = 1}^{k} |D_i|$ that is obtained by concatenating all tuples $D_j$, $1 \leq j \leq k$, in the order given by the Janus operating mode. For every $i$, $1 \leq i \leq k'$, let 
\begin{equation*}
s_i := |\{x~|~\exists~j, j' \mbox{ with } 1 \leq j < i < j' \leq k', y_{d'_j} = y_{d'_{j'}} = x \neq y_{d'_{i}}\}|\,.
\end{equation*}
Then the \emph{counter number of $\Delta_{\alpha}$} (or \emph{$\cn(\Delta_{\alpha})$} for short) is $\max\{s_i~|~1 \leq i \leq k'\}$.
\end{definition}

We explain the previous definition in an informal manner. Apart from the markers $\lambda$ and $\rho$, the head movement indicator $\overline{\Delta_{\alpha}}$, where $\Delta_{\alpha}$ is a Janus operating mode for some $\alpha$, can be regarded as a sequence $(d'_1, d'_2, \ldots, d'_{k'})$, where the $d'_i$, $1 \leq i \leq k'$, are positions in $\alpha$. Hence, we can associate a pattern $D_{\alpha} := y_{d'_1} \cdot y_{d'_2} \cdot \ldots \cdot y_{d'_{k'}}$ with $\overline{\Delta_{\alpha}}$. In order to determine the counter number of $\Delta_{\alpha}$, we consider each position $i$, $1 \leq i \leq k'$, in $D_{\alpha}$ and count the number of variables different from $y_{d'_{i}}$ that are parenthesising position $i$ in $D_{\alpha}$. The counter number is then the maximum over all these numbers.\par
With regard to our example $\beta$, it can be easily verified that $\cn(\Delta_{\beta}) = 2$. We shall now see that, for every Janus operating mode $\Delta_{\alpha}$ for a pattern $\alpha$, we can construct a Janus automaton recognising $L(\alpha)$ with exactly $\cn(\Delta_{\alpha}) + 1$ counters:

\begin{theorem}\label{janusConstructionTheorem}
Let $\alpha$ be a terminal-free pattern and let $\Delta_{\alpha}$ be an arbitrary Janus operating mode for $\alpha$. There exists a $\jfa(\cn(\Delta_{\alpha}) + 1)$ $M$ satisfying $L(M) = L(\alpha)$.
\end{theorem}

Before we can prove this result, we need the following technical lemma:

\begin{lemma} \label{crossingLemma}
Let $\alpha$ be a terminal-free pattern with $|\var(\alpha)| \geq 2$, and let $\Gamma := \{z_1, z_2, \ldots, z_m\} \subseteq \var(\alpha)$. The following statements are equivalent:
\renewcommand*\theenumi{\alph{enumi}}
\begin{enumerate}
\item \label{a_stat} For all $z, z' \in \Gamma$, $z \neq z'$, the pattern $\alpha$ can be factorised into $\alpha = \beta \cdot z \cdot \gamma \cdot z' \cdot \gamma' \cdot z \cdot \delta$ or $\alpha = \beta \cdot z' \cdot \gamma \cdot z \cdot \gamma' \cdot z' \cdot \delta$.
\item \label{b_stat} There exists a $z \in \Gamma$ such that $\alpha$ can be factorised into $\alpha = \beta \cdot z \cdot \gamma$ with $(\Gamma \slash \{z\}) \subseteq (\var(\beta) \cap \var(\gamma))$.
\end{enumerate}
\end{lemma}

\begin{proof}
We prove by contraposition that \ref{a_stat} implies \ref{b_stat}. Hence, we assume that there exists no $z \in \Gamma$ such that $\alpha$ can be factorised into $\alpha = \beta \cdot z \cdot \gamma$ with $(\Gamma \slash \{z\}) \subseteq (\var(\beta) \cap \var(\gamma))$. Next, we define $l_{1}, l_{2}, \ldots, l_{m}$ to be the leftmost occurrences and $r_{1}, r_{2}, \ldots, r_{m}$ to be the rightmost occurrences of the variables $z_{1}, z_{2}, \ldots, z_{m}$. Furthermore, we assume $l_{1} < l_{2} < \ldots < l_{m}$. By assumption, it is not possible that, for every $i$, $1 \leq i \leq m - 1$, $r_i > l_m$ as this implies that $\alpha$ can be factorised into $\alpha = \beta \cdot z_m \cdot \gamma$, $|\beta| = l_m - 1$ with $(\Gamma \slash \{z_m\}) \subseteq (\var(\beta) \cap \var(\gamma))$. So we can assume that there exists an $i$, $1 \leq i \leq m - 1$, with $r_i < l_m$. This implies that, for $z_i, z_m$, $\alpha$ can neither be factorised into $\alpha = \beta \cdot z_i \cdot \gamma \cdot z_m \cdot \gamma' \cdot z_i \cdot \delta$ nor into $\alpha = \beta \cdot z_m \cdot \gamma \cdot z_i \cdot \gamma' \cdot z_m \cdot \delta$. This proves that \ref{a_stat} implies \ref{b_stat}.\par
The converse statement, \ref{b_stat} implies \ref{a_stat}, can be easily comprehended. We assume that $z \in \Gamma$ satisfies the conditions of \ref{b_stat}, i.\,e., $\alpha$ can be factorised into $\alpha = \beta \cdot z \cdot \gamma$ with $(\Gamma \slash \{z\}) \subseteq (\var(\beta) \cap \var(\gamma))$. Now we arbitrarily choose $z', z'' \in \Gamma$, $z' \neq z''$, and we shall show that $\alpha = \beta' \cdot z' \cdot \gamma' \cdot z'' \cdot \gamma'' \cdot z' \cdot \delta'$ or $\alpha = \beta' \cdot z'' \cdot \gamma' \cdot z' \cdot \gamma'' \cdot z'' \cdot \delta'$. If either $z' = z$ or $z'' = z$, this is obviously true. In all other cases, the fact that there are occurrences of both $z'$ and $z''$ to either side of the occurrence of $z$ directly implies the existence of one of the aforementioned factorisations.
\end{proof}

Now we are able to present the proof of Theorem~\ref{janusConstructionTheorem}:

\begin{proof}
Let $\pi := \cn(\Delta_{\alpha}) + 1$. In order to prove Theorem~\ref{janusConstructionTheorem}, we illustrate a general way of transforming a Janus operating mode $\Delta_{\alpha} := (D_1, D_2, \ldots, D_k)$ of an arbitrary terminal-free pattern $\alpha := y_1 \cdot y_2 \cdot \ldots \cdot y_n$ into a Janus automaton $M$ with $\cn(\Delta_{\alpha}) + 1$ counters satisfying $L(M) = L(\alpha)$. We shall first give a definition of the automaton and then prove its correctness, i.\,e., $L(M) = L(\alpha)$.\par 
We assume that the Janus operating mode is derived from the complete matching order $(m_1, m_2, \ldots, m_k)$. Let us recall the main definitions that are used in this proof, namely the complete matching order and the Janus operating mode. We know that each element $m_i$, $1 \leq i \leq k$, of the complete matching order is a matching position, i.\,e., $m_i = (l_i, r_i)$, $l_i < r_i$ and $y_{l_i} = y_{r_i}$. The complete matching order is included in the Janus operating mode, since, for each $i$, $1 \leq i \leq k$, the tuple $D_i$ corresponds to the matching position $m_i$ in the following way: If $m_i = (l_i, r_i)$, then the last two elements of $D_i$ are $(r_i, \rho)$ and $(l_i, \lambda)$. All the other pairs in a $D_i$ are of form $(j, \mu)$ where $1 \leq j \leq |\alpha|$ and $\mu \in \{\lambda, \rho\}$. \par
Before we move on to the formal definitions of the states and transitions of the automaton, let us illustrate its behaviour in an informal way. As described at the beginning of Section~\ref{sec:pattern}, the membership problem can be solved by checking the existence of a characteristic factorisation $u_1 \cdot u_2 \cdot \ldots \cdot u_n$ of the input $w$. Furthermore, by Lemma~\ref{matchingOrderLemma}, the complete matching order can be used as a list of instructions to perform this task. The factorisation is defined by the counter bounds, i.\,e., for every variable $x \in \var(\alpha)$, the automaton uses a certain counter, the counter bound of which defines the length of all the factors $u_i$ with $y_i = x$. However, if $\pi < |\var(\alpha)|$ is satisfied, then the automaton does not have the number of counters required for such a representation. Therefore, it might be necessary to reuse counters. To define which counter is used for which variables, we use a mapping $\co : \var(\alpha) \rightarrow \{1, 2, \ldots, \pi\}$. Note that, in case of $\pi < |\var(\alpha)|$, this mapping is not injective. We defer a complete definition of the mapping $\co$ and, for now, just assume that there exists such a mapping. \par
Next, we show how a tuple $D_p$ for an arbitrary $p$, $1 \leq p \leq k$, can be transformed into a part of the automaton. Therefore, we define
\begin{equation*}
D_p := ((j_1, \mu_1), (j_2, \mu_2), \ldots, (j_{k'}, \mu_{k'}), (j_r, \rho), (j_l, \lambda))
\end{equation*}
with $\mu_i \in \{\lambda, \rho\}$, $1 \leq i \leq k'$. Recall that $D_p$ corresponds to the matching position $m_p := (j_l, j_r)$. Let us interpret the tuple $D_p$ as follows: The pairs $(j_1, \mu_1), (j_2, \mu_2), \ldots, (j_{k'}, \mu_{k'})$ define how the heads have to be moved in order to reach factors $u_{j_l}$ and $u_{j_r}$, which then have to be matched. Let $(j_i, \mu_i)$, $1 \leq i \leq k'$, be an arbitrary pair of $D_p$. If $\mu_i = \lambda$ (or $\mu_i = \rho$), then the meaning of this pair is that the left head (or the right head, respectively) has to be moved a number of steps defined by the counter bound of counter $\co(y_{j_i})$. The direction the head has to be moved to depends on the matching position corresponding to the previous element $D_{p - 1}$. In order to define these ideas formally, we refer to this previous matching position by $m_{p - 1} := (j'_l, r'_l)$.\par
If $j'_l < j_l$, then we have to move the left head to the right passing the factors $u_{j'_{l} + 1}, u_{j'_l + 2}, \ldots, u_{j_l - 1}$; thus, we introduce the following states:
\begin{equation*}
\{\lforth_{p, q} \mid j'_l + 1 \leq q \leq j_l - 1\}\enspace.
\end{equation*}
In every state $\lforth_{p, q}$, $j'_l + 1 \leq q \leq j_l - 1$, we move the left head as many steps to the right as determined by the currently stored counter bound for counter $\co(y_{q})$. Hence, for every $q$, $j'_l + 1 \leq q \leq j_l - 1$, for all $a, a' \in \Sigma$ and for every $s_i \in \{\mathtt{t_{=}}, \mathtt{t_{<}}\}$, $i \in \{1, \ldots, \pi\} \slash \{\co(y_{q})\}$, we define
\begin{equation*}
\delta(\lforth_{p, q}, a, a', s_1, s_2, \ldots, s_{\pi}) := (\lforth_{p, q}, 1, 0, d_1, d_2, \ldots, d_{\pi})\enspace,
\end{equation*}
where $s_{\co(y_{q})} := \mathtt{t_{<}}$, $d_{\co(y_{q})} := 1$, and, for every $i \in \{1, \ldots, \pi\} \slash \{\co(y_{q})\}$, $d_i := 0$.\par
Analogously, if $j_l < j'_l$, then we have to move the left head to the left over the factors $u_{j'_{l}}, u_{j'_{l} - 1}, \ldots, u_{j_{l} + 1}, u_{j_{l}}$; to this end we use the following set of states:
\begin{equation*}
\{\lback_{p, q} \mid j_{l} \leq q \leq j'_l\}\enspace.
\end{equation*}
As before, for every $q$, $j_{l} \leq q \leq j'_l$, for all $a, a' \in \Sigma$ and for every $s_i \in \{\mathtt{t_{=}}, \mathtt{t_{<}}\}$, $i \in \{1, \ldots, \pi\} \slash \{\co(y_{q})\}$, we define
\begin{equation*}
\delta(\lback_{p, q}, a, a', s_1, s_2, \ldots, s_{\pi}) := (\lback_{p, q}, -1, 0, d_1, d_2, \ldots, d_{\pi})\enspace,
\end{equation*}
where $s_{\co(y_{q})} := \mathtt{t_{<}}$, $d_{\co(y_{q})} := 1$, and, for every $i \in \{1, \ldots, \pi\} \slash \{\co(y_{q})\}$, $d_i := 0$.\par
Note that, in the above defined transitions, the only difference between the cases $j'_l < j_l$ and $j_l < j'_l$, apart from the different states, is the head instruction for the left head. The states for the right head, i.\,e., $\rforth_{p, q}$ and
$\rback_{p, q}$, and their transitions are defined analogously.\par
Up to now, we have introduced states that can move the input heads back or forth over whole factors of the input word. This is done by moving an input head and simultaneously incrementing a counter until it reaches the counter bound, i.\,e., the counter message changes to $\mathtt{t_{=}}$. It remains to define what happens if an input head is completely moved over a factor and the counter message changes to $\mathtt{t_{=}}$. Intuitively, in this case the automaton should change to another state and then move a head in dependency of another counter. Thus, e.\,g., if in state $\lforth_{p, i}$ the counter message of counter $\co(y_{i})$ is $\mathtt{t_{=}}$, then the automaton should change into state $\lforth_{p, i + 1}$. In order to simplify the formal definition we assume $j'_l < j_l$ and $j'_r < j_r$, as all other cases can be handled similarly. 
For every $q$, $1 \leq q \leq k' - 1$, for all $a, a' \in \Sigma$ and for every $s_i \in \{\mathtt{t_{=}}, \mathtt{t_{<}}\}$, $i \in \{1, \ldots, \pi\} \slash \{\co(y_{q})\}$, we define
\begin{align*}
&\delta(\lforth_{p, q}, a, a', s_1, s_2, \ldots, s_{\pi}) := (\lforth_{p, q + 1}, 0, 0, d_1, d_2, \ldots, d_{\pi})\enspace,\\
&\mbox{if $\mu_p = \lambda$ and $\mu_{p + 1} = \lambda$}\enspace,
\end{align*}
\begin{align*}
&\delta(\lforth_{p, q}, a, a', s_1, s_2, \ldots, s_{\pi}) := (\rforth_{p, q + 1}, 0, 0, d_1, d_2, \ldots, d_{\pi})\enspace,\\
&\mbox{if $\mu_p = \lambda$ and $\mu_{p + 1} = \rho$}\enspace,
\end{align*}
\begin{align*}
&\delta(\rforth_{p, q}, a, a', s_1, s_2, \ldots, s_{\pi}) := (\lforth_{p, q + 1}, 0, 0, d_1, d_2, \ldots, d_{\pi})\enspace,\\
&\mbox{if $\mu_p = \rho$ and $\mu_{p + 1} = \lambda$}\enspace,
\end{align*}
\begin{align*}
&\delta(\rforth_{p, q}, a, a', s_1, s_2, \ldots, s_{\pi}) := (\rforth_{p, q + 1}, 0, 0, d_1, d_2, \ldots, d_{\pi})\enspace,\\
&\mbox{if $\mu_p = \rho$ and $\mu_{p + 1} = \rho$}\enspace,
\end{align*}
where $s_{\co(y_{q})} := \mathtt{t_{=}}$, $d_{\co(y_{q})} = 1$, and, for every $i \in \{1, \ldots, \pi\} \slash \{\co(y_{q})\}$, $d_i := 0$.\par
Now, for every $i$, $1 \leq i \leq k' - 1$, the transition changing the automaton from the state corresponding to the pair $(j_i, \mu_i)$ into the state corresponding to $(j_{i + 1}, \mu_{i + 1})$ has been defined. Note, that in these transitions we increment the counter $\co(y_{q})$ once more without moving the input head to set its value back to $0$ again, such that it is ready for the next time it is used. However, it remains to define what happens if the counter $\co(y_{j_{k'}})$ reaches its counter bound in the state that corresponds to the final pair $(j_{k'}, \mu_{k'})$. In this case, the automaton enters a new state $\match_p$, in which the factors $u_{j_l}$ and $u_{j_r}$ are matched. 
In the following definition, let $q := j_{k'}$. For all $a, a' \in \Sigma$ and for every $s_i \in \{\mathtt{t_{=}}, \mathtt{t_{<}}\}$, $i \in \{1, \ldots, \pi\} \slash \{\co(y_{q})\}$, we define
\begin{align*}
&\delta(\lforth_{p, q}, a, a', s_1, s_2, \ldots, s_{\pi}) := (\match_{p}, 0, 0, d_1, d_2, \ldots, d_{\pi})\enspace,\\
&\mbox{if $\mu_{j_{k'}} = \lambda$}\enspace,
\end{align*}
\begin{align*}
&\delta(\rforth_{p, q}, a, a', s_1, s_2, \ldots, s_{\pi}) := (\match_{p}, 0, 0, d_1, d_2, \ldots, d_{\pi})\enspace,\\
&\mbox{if $\mu_{j_{k'}} = \rho$}\enspace,
\end{align*}
where $s_{\co(y_{q})} := \mathtt{t_{=}}$, $d_{\co(y_{q})} := 1$, and, for every $i \in \{1, \ldots, \pi\} \slash \{\co(y_{q})\}$, $d_i := 0$.\par
In the state $\match_{p}$ the factors $u_{j_l}$ and $u_{j_r}$ are matched by simultaneously moving both heads to the right. 
In the following definition, let $q := j_l$. For every $a \in \Sigma$ and for every $s_i \in \{\mathtt{t_{=}}, \mathtt{t_{<}}\}$, $i \in \{1, \ldots, {\pi}\} \slash \{\co(y_{q})\}$, we define
\begin{equation*}
\delta(\match_{p}, a, a, s_1, s_2, \ldots, s_{\pi}) := (\match_{p}, 1, 1, d_1, d_2, \ldots, d_{\pi})\enspace,
\end{equation*}
where $s_{\co(y_{q})} := \mathtt{t_{<}}$, $d_{\co(y_{q})} := 1$, and, for every $i \in \{1, \ldots, {\pi}\} \slash \{\co(y_{q})\}$, $d_i := 0$.\par
Note, that these transitions are only applicable if both input heads scan the same symbol. If the symbol scanned by the left head differs from the one scanned by the right head, then no transition is defined and thus the automaton stops in a non-accepting state.\par
Finally, the very last transition to define in order to transform $D_p$ into a part of the automaton is the case when counter $\co(y_{j_l})$ has reached its counter bound in state $\match_{p}$. For the sake of convenience, we assume that the first pair of $D_{p + 1}$ is $(j', \lambda)$ and, furthermore, that $m_{p + 1} := (j''_l, j''_r)$ with $j_l < j''_l$.
For all $a, a' \in \Sigma$ and for every $s_i \in \{\mathtt{t_{=}}, \mathtt{t_{<}}\}$, $i \in \{1, \ldots, \pi\} \slash \{\co(y_{q})\}$, we define
\begin{equation*}
\delta(\match_{p}, a, a', s_1, s_2, \ldots, s_{\pi}) := (\lforth_{p+1, j'}, 0, 0, d_1, d_2, \ldots, d_{\pi})\enspace,
\end{equation*}
where $s_{\co(y_{q})} := \mathtt{t_{=}}$, $d_{\co(y_{q})} := 1$, and, for every $i \in \{1, \ldots, \pi\} \slash \{\co(y_{q})\}$, $d_i := 0$.\par
As mentioned above, this is merely the transition in the case that the first pair of $D_{p + 1}$ is $(j', \lambda)$ and $j_l < j''_l$ is satisfied. However, all the other cases can be handled analogously. In the case that the first pair of $D_{p + 1}$ is $(j', \rho)$ instead of $(j', \lambda)$ we have to enter state $\rforth_{p+1, j'}$ instead of $\lforth_{p+1, j'}$. If $j_l > j''_l$ holds instead of $j_l < j''_l$ we have to enter a back-state (e.\,g., $\lback_{p+1, j'}$) instead. These transitions can also be interpreted as the passage between the part of the automaton corresponding to $D_p$ and the part corresponding to the next tuple $D_{p + 1}$ of the Janus operating mode. \par
We have to explain a few special cases concerning the definitions above. Regarding the tuples $D_1$ and $D_k$ we have to slightly change the definitions. Initially, both heads are located at the very left position of the input, i.\,e., the left endmarker ``$\cent$'', therefore only $\lforth_{1, q}$ and $\rforth_{1, q}$ states are needed to transform $D_1$ into a part of the automaton. When the automaton is in state $\match_{k}$ and the counter has reached its counter bound, then the state $q_f$ is entered, which is the only final state of $M$. We recall, that $\alpha = y_1 \cdot y_2 \cdot \ldots \cdot y_n$. Whenever the automaton, for a $p$, $1 \leq p \leq k$, is in a state in  $\{\lforth_{p, n}, \lback_{p, n}, \rforth_{p, n}, \rback_{p, n}\}$ or in a state $\match_{p}$, where $m_{p} = (j, n)$, for some $j$, $j < n$, is a matching position, then this means that a head is moved over the rightmost factor $u_n$. When the automaton is in such a state for the first time and the counter bound of counter $\co(y_n)$ is reached, then the automaton blocks if the head does not scan the right endmarker ``$\$$'', as this implies $|u_1 \cdot u_2 \cdot \ldots \cdot u_n| < |w|$. In case that $|u_1 \cdot u_2 \cdot \ldots \cdot u_n| > |w|$ the automaton blocks at some point when it tries to move a head to the right that scans $\$$ since this transition is not defined. A formal definition of these special cases is omitted.\par
Obviously, each of the above defined transitions depend on a certain counter determined by the mapping $\co$, so let us now return to the problem of defining this mapping. As already mentioned, this mapping $\co$ is in general not injective, hence it is possible that $\co(x) = \co(z)$ for some $x \neq z$. This means, intuitively speaking, that there seems to be an undesirable connection between the lengths of factors $u_j$ with $y_j = x$ and factors $u_{j'}$ with $y_{j'} = z$. However, this connection does not have any effect if it is possible to, initially, exclusively use the counter bound of counter $\co(x) = \co(z)$ for factors corresponding to $x$ and then exclusively for factors corresponding to variable $z$ and never for factors corresponding to $x$ again. In this case the automaton may reset this counter after it has been used for factors corresponding to $x$ in order to obtain a new length for factors corresponding to $z$. This means that a counter is reused. We now formalise this idea.\par
Let $\overline{\Delta_{\alpha}} := ((d'_1, \mu'_1), (d'_2, \mu'_2), \ldots, (d'_{k''}, \mu'_{k''}))$ be the head movement indicator of the Janus operating mode. We consider the pattern $D_{\alpha} := y_{d'_1} \cdot y_{d'_2} \cdot \ldots \cdot y_{d'_{k''}}$. If, for some $x, z \in \var(\alpha)$, $x \neq z$, $D_{\alpha}$ can be factorised into $D_{\alpha} = \beta \cdot x \cdot \gamma \cdot z \cdot \gamma' \cdot x \cdot \delta$, then the automaton cannot use the same counter for variables $x$ and $z$; thus, $\co$ has to satisfy $\co(x) \neq \co(z)$.\par\bigskip\noindent
\emph{Claim} There exists a total mapping $\co : \var(\alpha) \rightarrow \{1, 2, \ldots, \pi\}$ such that, for all $x, z \in \var(\alpha)$, $x \neq z$, if $D_{\alpha} = \beta \cdot x \cdot \gamma \cdot z \cdot \gamma' \cdot x \cdot \delta$ or $D_{\alpha} = \beta \cdot z \cdot \gamma \cdot x \cdot \gamma' \cdot z \cdot \delta$, then $\co(x) \neq \co(z)$.\par\medskip\noindent
\emph{Proof (Claim).} If there is no set of variables $\Gamma \subseteq \var(\alpha)$ with $|\Gamma| > \pi$ such that for all $x, z \in \Gamma$, $x \neq z$, $D_{\alpha} = \beta \cdot x \cdot \gamma \cdot z \cdot \gamma' \cdot x \cdot \delta$ or $D_{\alpha} = \beta \cdot z \cdot \gamma \cdot x \cdot \gamma' \cdot z \cdot \delta$, then there obviously exists such a mapping $\co$. So we assume to the contrary, that there exists a set of variables $\Gamma$, $|\Gamma| = \pi + 1$, with the above given properties. Now we can apply Lemma~\ref{crossingLemma} to the pattern $D_{\alpha}$ and conclude that there exist a $z' \in \Gamma$ such that $D_{\alpha}$ can be factorised into $D_{\alpha} = \beta \cdot z' \cdot \gamma$ with $(\Gamma \slash \{z'\}) \subseteq (\var(\beta) \cap \var(\gamma))$. This directly implies $\cn(\Delta_{\alpha}) \geq \pi = \cn(\Delta_{\alpha}) + 1$, which is a contradiction.\par \hfill \emph{q.e.d. (Claim)} \par\bigskip\noindent
This shows that such a mapping $\co$ exists and, furthermore, we can note that it is straightforward to effectively construct it.\par
As already mentioned above, it may be necessary for the automaton to reset counters. More formally, if, for some $j$, $1 \leq j \leq \pi$, and for some $x, z \in \var(\alpha)$, $x \neq z$, $\co(x) = \co(z) = j$, then this counter $j$ must be reset. We now explain how this is done. By definition of the states and transitions so far, we may interpret states as being related to factors $u_q$, i.\,e., for every $p$, $1 \leq p \leq k$, and every $q$, $1 \leq q \leq n$, the states in $\{\lforth_{p, q}, \lback_{p, q}, \rforth_{p, q}, \rback_{p, q}\}$ correspond to factor $u_q$ and state $\match_{p}$ corresponds to both factors $u_l$ and $u_r$, where $m_p = (l, r)$. For every $x \in \var(\alpha)$, the automaton resets counter $\co(x)$, using the special counter instruction $\mathtt{r}$, immediately after leaving the last state corresponding to a factor $u_q$ with $y_q = x$. In order to define this transition formally, we assume that, for example, $\lforth_{p, q}$ with $y_q = x$ is that state and $\lforth_{p, q + 1}$ is the subsequent state.
For all $a, a' \in \Sigma$ and for every $s_i \in \{\mathtt{t_{=}}, \mathtt{t_{<}}\}$, $i \in \{1, \ldots, \pi\} \slash \{\co(x)\}$, we define
\begin{equation*}
\delta(\lforth_{p, q}, a, a', s_1, s_2, \ldots, s_{\pi}) = (\lforth_{p, q + 1}, 0, 0, d_1, d_2, \ldots, d_{\pi})\enspace,
\end{equation*}
where $s_{\co(x)} := \mathtt{t_{=}}$, $d_{\co(x)} := \mathtt{r}$, and, for every $i \in \{1, \ldots, \pi\} \slash \{\co(x)\}$, $d_i := 0$.\par
We recall, that by definition of a Janus automaton, all counter bounds are initially $0$, so the automaton must initially reset all $\pi$ counters. To define this transition formally, let $\lforth_{1,1}$ be the state corresponding to the first element of $D_1$. The first transition is defined by
\begin{equation*}
\delta(q_0, \cent, \cent, \mathtt{t_{=}}, \mathtt{t_{=}}, \ldots, \mathtt{t_{=}}) = (\lforth_{1,1}, 0, 0, \mathtt{r}, \mathtt{r}, \ldots, \mathtt{r})\enspace,
\end{equation*}
where $q_0$ is the initial state of $M$. This concludes the definition of the automaton and we shall now prove its correctness, i.\,e., $L(M) = L(\alpha)$.\par
Let $w \in \Sigma^*$ be an arbitrary input word. From the above given definition, it is obvious that the automaton treats $w$ as a sequence of factors $u_1 \cdot u_2 \cdot \ldots \cdot u_n$. The lengths of these factors $u_i$, $1 \leq i \leq n$, are determined by the counter bounds guessed during the computation. If $|u_1 \cdot u_2 \cdot \ldots \cdot u_n| \neq |w|$, then the automaton does not accept the input anyway, so we may only consider those cases where suitable counter bounds are guessed that imply $|u_1 \cdot u_2 \cdot \ldots \cdot u_n| = |w|$. Recall the complete matching order $(m_1, m_2, \ldots, m_k)$ with $m_p = (l_p, r_p)$, $1 \leq p \leq k$. By definition, in the states $\match_p$, $1 \leq p \leq k$, the automaton matches factor $u_{l_p}$ and $u_{r_p}$. If $M$ reaches the accepting state $q_f$, then, for every $p$, $1 \leq p \leq k$, $u_{l_p} = u_{r_p}$ and, by applying Lemma~\ref{matchingOrderLemma}, we conclude that $u_1 \cdot u_2 \cdot \ldots \cdot u_n$ is a characteristic factorisation. Hence, $w \in L(\alpha)$.\par
On the other hand, let $w' \in L(\alpha)$ be arbitrarily chosen. This implies that we can factorise $w'$ into $w' = u_1 \cdot u_2 \cdot \ldots \cdot u_n$ such that for all $j, j'$, $1 \leq j < j' \leq n$, $y_j = y_{j'}$ implies $u_j = u_{j'}$, i.\,e., $u_1 \cdot u_2 \cdot \ldots \cdot u_n$ is a characteristic factorisation. By definition, it is possible that the automaton guesses counter bounds such that the input word $w'$ is treated in this factorisation $w' = u_1 \cdot u_2 \cdot \ldots \cdot u_n$, so $M$ accepts $w'$ and thus $w' \in L(M)$. Consequently, $L(M) = L(\alpha)$, which concludes the proof of correctness, and hence the proof of Theorem~\ref{janusConstructionTheorem}.
\end{proof}

We conclude this section by discussing the previous results in a bit more detail. The main technical tool defined in this section is the Janus operating mode. So far, we interpreted Janus operating modes as instructions specifying how two input heads can be used to move over a word given in a certain factorisation in order to check on whether this factorisation is a characteristic one. So, in other words, a Janus operating mode can be seen as representing an algorithm, solving the membership problem for the pattern language given by a certain pattern. Theorem~\ref{janusConstructionTheorem} formally proves this statement.\par
A major benefit of this approach is, that from now on we can focus on Janus operating modes rather than on the more involved model of a Janus automaton. More precisely, the previous result shows that the task of finding an optimal Janus automaton for a terminal-free pattern language is equivalent to finding an optimal Janus operating mode for this pattern. Before we investigate this task in the subsequent section, we revise our perspective regarding Janus operating modes. There is no need to consider input words anymore and, thus, in the following we shall investigate properties of patterns and Janus operating modes exclusively. Therefore, we establish a slightly different point of view at Janus operating modes, i.\,e., we interpret them as describing input head movements over a pattern instead of over a word given in  a factorisation:

\begin{remark}\label{inputHeadsPerspectiveRemark}
Let $\Delta_{\alpha} := (D_1, D_2, \ldots, D_k)$ be an arbitrary Janus operating mode for some pattern $\alpha := y_1 \cdot y_2 \cdot \ldots \cdot y_n$ and let $\Delta_{\alpha}$ be derived from the complete matching order $(m_1, m_2, \ldots, m_k)$. Furthermore, let $\overline{\Delta_{\alpha}} := ((d'_1, \mu'_1)$, $(d'_2, \mu'_2), \ldots, (d'_{k'}, \mu'_{k'}))$ be the head movement indicator of the canonical Janus operating mode. We can interpret $\overline{\Delta_{\alpha}}$ as a sequence of input head movements over the pattern $\alpha$, i.\,e., after $i$ movements or steps of $\overline{\Delta_{\alpha}}$, where $1 \leq i \leq k'$, the left input head is located at variable $y_{d'_i}$ if $\mu'_i = \lambda$ or, in case that $\mu'_i = \rho$, the right input head is located at $y_{d'_i}$. So for every $i$, $1 \leq i \leq k'$, the sequence $\overline{\Delta_{\alpha}}$ determines the positions of both input heads after the first $i$ movements of $\overline{\Delta_{\alpha}}$. More precisely, for every $i$, $1 \leq i \leq k'$, after $i$ steps of $\overline{\Delta_{\alpha}}$, the positions $l_i$ and $r_i$ of the left head and the right head in $\alpha$ are given by 
\begin{align*}
l_i &= \max\{d'_j \mid 1 \leq j \leq i, \mu'_j = \lambda\} \text{ and}\\
r_i &= \max\{d'_j \mid 1 \leq j \leq i, \mu'_j = \rho\}\,.
\end{align*}
We note that $\{d'_j \mid 1 \leq j \leq i, \mu'_j = \lambda\} = \emptyset$ is possible, which means that $\mu_j = \rho$, $1 \leq j \leq i$, or, in other words, that so far only the right head has been moved. In this case, we shall say that the left head has not yet entered $\alpha$ and therefore is located at position $0$. The situation $\{d'_j \mid 1 \leq j \leq i, \mu'_j = \rho\} = \emptyset$ is interpreted analogously. As already mentioned above, for every $i$, $1 \leq i \leq k'$, we have either $l_i = d'_i$ or $r_i = d'_i$ (depending on $\mu_i$). Furthermore, for every $i$, $1 \leq i \leq k'$, it is not possible that both heads are located at position $0$.
\end{remark}

This special perspective towards Janus operating modes, described in the previous remark, shall play a central role in the proofs for the following results.

\section{Patterns with Restricted Variable Distance}\label{sec:varDist}

We now introduce a certain combinatorial property of terminal-free patterns, the so-called variable distance, which is the maximum number of different variables separating any two consecutive occurrences of a variable:
\begin{definition}\label{vdDefinition}
The \emph{variable distance} of a terminal-free pattern $\alpha$ is the smallest number $k \geq 0$ such that, for every $x \in \var(\alpha)$, every factorisation $\alpha = \beta \cdot x \cdot \gamma \cdot x \cdot \delta$ with $\beta, \gamma, \delta \in \var(\alpha)^*$ and $|\gamma|_{x} = 0$ satisfies $|\var(\gamma)| \leq k$. We denote the variable distance of a terminal-free pattern $\alpha$ by $\vd(\alpha)$.
\end{definition}

Obviously, $\vd(\alpha) \leq |\var(\alpha)| - 1$ for all terminal-free patterns $\alpha$. To illustrate the concept of the variable distance, we consider the pattern $\beta' := x_1 \cdot x_2 \cdot x_3 \cdot x_2 \cdot x_3 \cdot x_1 \cdot x_4 \cdot x_3 \cdot x_5 \cdot x_5 \cdot x_4$. In the following figure, for every two successive occurrences of any variable in $\beta'$, the number of different variables occurring between these occurrences is pointed out:

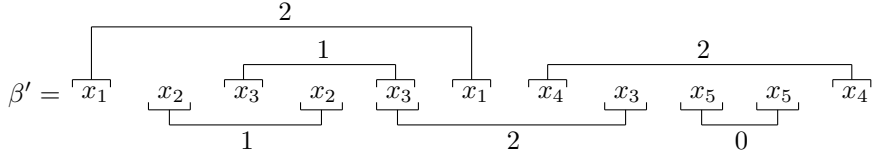
\begin{tikzpicture}

\coordinate (coord0) at (0,0);
\coordinate (coord1) at ($(coord0) + (1,0)$);
\coordinate (coord2) at ($(coord1) + (1,0)$);
\coordinate (coord3) at ($(coord2) + (1,0)$);
\coordinate (coord4) at ($(coord3) + (1,0)$);
\coordinate (coord5) at ($(coord4) + (1,0)$);
\coordinate (coord6) at ($(coord5) + (1,0)$);
\coordinate (coord7) at ($(coord6) + (1,0)$);
\coordinate (coord8) at ($(coord7) + (1,0)$);
\coordinate (coord9) at ($(coord8) + (1,0)$);
\coordinate (coord10) at ($(coord9) + (1,0)$);

\draw ($(coord0) - (0.8,0)$) node {$\beta' = $};

\draw (coord0) node {$x_1$};
\draw (coord1) node {$x_2$};
\draw (coord2) node {$x_3$};
\draw (coord3) node {$x_2$};
\draw (coord4) node {$x_3$};
\draw (coord5) node {$x_1$};
\draw (coord6) node {$x_4$};
\draw (coord7) node {$x_3$};
\draw (coord8) node {$x_5$};
\draw (coord9) node {$x_5$};
\draw (coord10) node {$x_4$};

\draw[black] ($(coord0) + (-0.05,0.9)$) -- ($(coord0) + (-0.05,0.2)$);
\draw[black] ($(coord0) + (-0.30,0.2)$) -- ($(coord0) + (0.2,0.2)$);
\draw[black] ($(coord0) + (-0.30,0.2)$) -- ($(coord0) + (-0.30,0.1)$);
\draw[black] ($(coord0) + (0.2,0.2)$) -- ($(coord0) + (0.2,0.1)$);
\draw[black] ($(coord5) + (-0.05,0.9)$) -- ($(coord5) + (-0.05,0.2)$);
\draw[black] ($(coord5) + (-0.30,0.2)$) -- ($(coord5) + (0.2,0.2)$);
\draw[black] ($(coord5) + (-0.30,0.2)$) -- ($(coord5) + (-0.30,0.1)$);
\draw[black] ($(coord5) + (0.2,0.2)$) -- ($(coord5) + (0.2,0.1)$);
\draw[black] ($(coord0) + (-0.05,0.9)$) -- ($(coord5) + (-0.05,0.9)$);
\draw[black] ($(coord2) + (0.5,1.1)$) node {$2$};

\draw[black] ($(coord2) + (-0.05,0.4)$) -- ($(coord2) + (-0.05,0.2)$);
\draw[black] ($(coord2) + (-0.30,0.2)$) -- ($(coord2) + (0.2,0.2)$);
\draw[black] ($(coord2) + (-0.30,0.2)$) -- ($(coord2) + (-0.30,0.1)$);
\draw[black] ($(coord2) + (0.2,0.2)$) -- ($(coord2) + (0.2,0.1)$);
\draw[black] ($(coord4) + (-0.05,0.4)$) -- ($(coord4) + (-0.05,0.2)$);
\draw[black] ($(coord4) + (-0.30,0.2)$) -- ($(coord4) + (0.2,0.2)$);
\draw[black] ($(coord4) + (-0.30,0.2)$) -- ($(coord4) + (-0.30,0.1)$);
\draw[black] ($(coord4) + (0.2,0.2)$) -- ($(coord4) + (0.2,0.1)$);
\draw[black] ($(coord2) + (-0.05,0.4)$) -- ($(coord4) + (-0.05,0.4)$);
\draw[black] ($(coord3) + (0,0.6)$) node {$1$};

\draw[black] ($(coord6) + (-0.05,0.4)$) -- ($(coord6) + (-0.05,0.2)$);
\draw[black] ($(coord6) + (-0.30,0.2)$) -- ($(coord6) + (0.2,0.2)$);
\draw[black] ($(coord6) + (-0.30,0.2)$) -- ($(coord6) + (-0.30,0.1)$);
\draw[black] ($(coord6) + (0.2,0.2)$) -- ($(coord6) + (0.2,0.1)$);
\draw[black] ($(coord10) + (-0.05,0.4)$) -- ($(coord10) + (-0.05,0.2)$);
\draw[black] ($(coord10) + (-0.30,0.2)$) -- ($(coord10) + (0.2,0.2)$);
\draw[black] ($(coord10) + (-0.30,0.2)$) -- ($(coord10) + (-0.30,0.1)$);
\draw[black] ($(coord10) + (0.2,0.2)$) -- ($(coord10) + (0.2,0.1)$);
\draw[black] ($(coord6) + (-0.05,0.4)$) -- ($(coord10) + (-0.05,0.4)$);
\draw[black] ($(coord8) + (0,0.6)$) node {$2$};

\draw[black] ($(coord1) + (-0.025,-0.4)$) -- ($(coord1) + (-0.025,-0.2)$);
\draw[black] ($(coord1) + (-0.3,-0.2)$) -- ($(coord1) + (0.25,-0.2)$);
\draw[black] ($(coord1) + (-0.3,-0.2)$) -- ($(coord1) + (-0.3,-0.1)$);
\draw[black] ($(coord1) + (0.25,-0.2)$) -- ($(coord1) + (0.25,-0.1)$);
\draw[black] ($(coord3) + (-0.025,-0.4)$) -- ($(coord3) + (-0.025,-0.2)$);
\draw[black] ($(coord3) + (-0.3,-0.2)$) -- ($(coord3) + (0.25,-0.2)$);
\draw[black] ($(coord3) + (-0.3,-0.2)$) -- ($(coord3) + (-0.3,-0.1)$);
\draw[black] ($(coord3) + (0.25,-0.2)$) -- ($(coord3) + (0.25,-0.1)$);
\draw[black] ($(coord1) + (-0.025,-0.4)$) -- ($(coord3) + (-0.025,-0.4)$);
\draw[black] ($(coord2) + (0,-0.6)$) node {$1$};

\draw[black] ($(coord4) + (-0.025,-0.4)$) -- ($(coord4) + (-0.025,-0.2)$);
\draw[black] ($(coord4) + (-0.3,-0.2)$) -- ($(coord4) + (0.25,-0.2)$);
\draw[black] ($(coord4) + (-0.3,-0.2)$) -- ($(coord4) + (-0.3,-0.1)$);
\draw[black] ($(coord4) + (0.25,-0.2)$) -- ($(coord4) + (0.25,-0.1)$);
\draw[black] ($(coord7) + (-0.025,-0.4)$) -- ($(coord7) + (-0.025,-0.2)$);
\draw[black] ($(coord7) + (-0.3,-0.2)$) -- ($(coord7) + (0.25,-0.2)$);
\draw[black] ($(coord7) + (-0.3,-0.2)$) -- ($(coord7) + (-0.3,-0.1)$);
\draw[black] ($(coord7) + (0.25,-0.2)$) -- ($(coord7) + (0.25,-0.1)$);
\draw[black] ($(coord4) + (-0.025,-0.4)$) -- ($(coord7) + (-0.025,-0.4)$);
\draw[black] ($(coord5) + (0.5,-0.6)$) node {$2$};

\draw[black] ($(coord8) + (-0.025,-0.4)$) -- ($(coord8) + (-0.025,-0.2)$);
\draw[black] ($(coord8) + (-0.3,-0.2)$) -- ($(coord8) + (0.25,-0.2)$);
\draw[black] ($(coord8) + (-0.3,-0.2)$) -- ($(coord8) + (-0.3,-0.1)$);
\draw[black] ($(coord8) + (0.25,-0.2)$) -- ($(coord8) + (0.25,-0.1)$);
\draw[black] ($(coord9) + (-0.025,-0.4)$) -- ($(coord9) + (-0.025,-0.2)$);
\draw[black] ($(coord9) + (-0.3,-0.2)$) -- ($(coord9) + (0.25,-0.2)$);
\draw[black] ($(coord9) + (-0.3,-0.2)$) -- ($(coord9) + (-0.3,-0.1)$);
\draw[black] ($(coord9) + (0.25,-0.2)$) -- ($(coord9) + (0.25,-0.1)$);
\draw[black] ($(coord8) + (-0.025,-0.4)$) -- ($(coord9) + (-0.025,-0.4)$);
\draw[black] ($(coord8) + (0.5,-0.6)$) node {$0$};
\end{tikzpicture}\\
Referring to the previous figure it can be easily comprehended that $\vd(\beta') = 2$.\par
The problem of computing the variable distance $\vd(\alpha)$ for an arbitrary pattern $\alpha$ is not a difficult one as pointed out by the following proposition: 
\begin{proposition}\label{vdEfficiencyProp}
For every terminal-free pattern $\alpha$, the number $\vd(\alpha)$ can be efficiently computed.
\end{proposition}

\begin{proof}
Let $\alpha := y_1 \cdot y_2 \cdot \ldots \cdot y_n$ be a terminal-free pattern. It is possible to compute the variable distance of $\alpha$ in the following way. We move over $\alpha$ from left to right. Whenever a variable $x$ is encountered for the first time, we initialise a set $S_x$, which we delete again after passing the last occurrence of $x$. Furthermore, for every $x$ that we pass, we add $x$ to all existing sets $S_{x\rq{}}$, $x \neq x\rq{}$, and completely empty the set $S_x$. The variable distance is then the maximum cardinality of any of these sets during this procedure.
\end{proof}

The following vital result shows that for every possible Janus operating mode for some pattern $\alpha$, its counter number is at least equal to the variable distance of $\alpha$. Hence, the variable distance is a lower bound for the counter number of Janus operating modes. 

\begin{theorem}\label{cnGeqVd}
Let $\Delta_{\alpha}$ be an arbitrary Janus operating mode for a terminal-free pattern $\alpha$. Then $\cn(\Delta_{\alpha}) \geq \vd(\alpha)$.
\end{theorem}

\begin{proof}
Let $\alpha := y_1 \cdot y_2 \cdot \ldots \cdot y_n$ be a terminal-free pattern and let $(m_1, m_2, \ldots, m_k)$ be the complete matching order for $\alpha$ from which $\Delta_{\alpha} := (D_1, D_2, \ldots, D_{k})$ is derived. Furthermore, let $\overline{\Delta_{\alpha}} := ((d'_1, \mu'_1), (d'_2, \mu'_2), \ldots, (d'_{k'}, \mu'_{k'}))$ be the head movement indicator of the Janus operating mode. This sequence $\overline{\Delta_{\alpha}}$ contains numbers $d'_i$, $1 \leq i \leq k'$, that are positions of $\alpha$, i.\,e., $1 \leq d'_i \leq |\alpha|$, $1 \leq i \leq k'$. Hence, we can associate a pattern $D_{\alpha}$ with $\overline{\Delta_{\alpha}}$ and $\alpha$ in the following way: $D_{\alpha} := y_{d'_1} \cdot y_{d'_2} \cdot \ldots \cdot y_{d'_{k'}}$. By definition of the variable distance, we know that there exists an $x \in \var(\alpha)$ such that $\alpha = \beta \cdot x \cdot \gamma \cdot x \cdot \delta$ with $|\gamma|_{x} = 0$ and $|\var(\gamma)| = \vd(\alpha)$. We assume $\vd(\alpha) \geq 1$ (i.\,e., $\var(\gamma) \neq \emptyset$), as in the case $\vd(\alpha) = 0$, $\cn(\Delta_{\alpha}) \geq \vd(\alpha)$ trivially holds.\par
In the following, let $\Gamma := \var(\gamma) \cup \{x\}$. We shall prove the statement of the theorem by showing that there exists a variable $z \in \Gamma$ such that $D_{\alpha} = \overline{\beta} \cdot z \cdot \overline{\gamma}$ with $|(\var(\overline{\beta}) \cap \var(\overline{\gamma})) \slash \{z\}| \geq \vd(\alpha)$, which implies $\cn(\Delta_{\alpha}) \geq \vd(\alpha)$. To this end, we first prove the following claim:
\par\bigskip\noindent
\emph{Claim} For all $z, z' \in \Gamma$, $z \neq z'$, we can factorise $D_{\alpha}$ into $D_{\alpha} = \widetilde{\beta} \cdot z \cdot \widetilde{\gamma_1} \cdot z' \cdot \widetilde{\gamma_2} \cdot z \cdot \widetilde{\delta}$ or $D_{\alpha} = \widetilde{\beta} \cdot z' \cdot \widetilde{\gamma_1} \cdot z \cdot \widetilde{\gamma_2} \cdot z' \cdot \widetilde{\delta}$.
\par\medskip\noindent
\emph{Proof (Claim).} 
For arbitrary $z, z' \in \Gamma$, $z \neq z'$, there are two possible cases regarding the positions of the occurrences of $z$ and $z'$ in $\alpha$. The first case describes the situation that there exists an occurrence of $z'$ (or $z$) in $\alpha$ such that $z$ (or $z'$, respectively) occurs to the left and to the right of this occurrence. If this is not possible, the occurrences of $z$ and $z'$ are separated, i.\,e., the rightmost occurrence of $z$ (or $z'$) is to the left of the leftmost occurrence of $z'$ (or $z$, respectively). More formally, it is possible to factorise $\alpha$ into
\begin{equation}\label{equ1}
\alpha = \widehat{\beta} \cdot z \cdot \widehat{\gamma_1} \cdot z' \cdot \widehat{\gamma_2} \cdot z \cdot \widehat{\delta}
\end{equation}
or into
\begin{equation}\label{equ2}
\alpha = \beta \cdot x \cdot \widehat{\gamma_1} \cdot z \cdot \widehat{\gamma_2} \cdot z' \cdot \widehat{\gamma_3} \cdot x \cdot \delta
\end{equation}
with $|\beta \cdot x \cdot \widehat{\gamma_1} \cdot z \cdot \widehat{\gamma_2}|_{z'} = 0$ and $|\widehat{\gamma_2} \cdot z' \cdot \widehat{\gamma_3} \cdot x \cdot \delta|_z = 0$. The two factorisations obtained by changing the roles of $z$ and $z'$ can be handeled analogously and are, thus, omitted. We note that in the second factorisation, $\widehat{\gamma_1} \cdot z \cdot \widehat{\gamma_2} \cdot z' \cdot \widehat{\gamma_3}$ equals the factor $\gamma$ from the above introduced factorisation $\alpha = \beta \cdot x \cdot \gamma \cdot x \cdot \delta$. This is due to the fact that we assume $z, z' \in \Gamma$.\par
We first observe that $z = x$ or $z' = x$ implies that the first factorisation is possible. If we cannot factorise $\alpha$ according to factorisation (\ref{equ1}), then we can conclude that the rightmost occurrence of $z$ is to the left of the leftmost occurrence of $z'$ and, furthermore, as both $z, z' \in \Gamma$ and $z \neq x \neq z'$, these occurrences are both in the factor $\gamma$. Hence, factorisation (\ref{equ2}) applies. We now show that in both cases the variables $z, z'$ satisfy the property described in the Claim. However, throughout the following argumentations, we need to bear in mind that the claim made above describes a property of $D_{\alpha}$ and the two considered factorisations are factorisations of $\alpha$.\par
We start with the case that $\alpha$ can be factorised into $\alpha = \widehat{\beta} \cdot z \cdot \widehat{\gamma_1} \cdot z' \cdot \widehat{\gamma_2} \cdot z \cdot \widehat{\delta}$. Let $p := |\widehat{\beta} \cdot z \cdot \widehat{\gamma_1} \cdot z' \cdot \widehat{\gamma_2}| + 1$, thus $y_{p} = z$. In the complete matching order $(m_1, \ldots, m_k)$ there has to be an $m_q$, $1 \leq q \leq k$, with $m_q := (j_l, j_r)$ and either $j_l = p$ or $j_r = p$. We assume that $j_l = p$; the case $j_r = p$ can be handled analogously. This implies, by definition of Janus operating modes, that the last element of $D_q$ is $(p, \lambda)$.\par
In the following, we interpret the Janus operating mode as a sequence of input head movements over $\alpha$, as explained in Remark~\ref{inputHeadsPerspectiveRemark}. Both heads start at the very left position of the input, so in order to move the left head to position $p$ in the pattern, it has to pass the whole part to the left of position $p$, i.\,e.\ $y_1 \cdot y_2 \cdot \ldots y_{p - 1}$, from left to right (possibly changing directions several times). In this initial part of the pattern, the variables $z$ and $z'$ occur in exactly this order. We conclude that the left head has to pass an occurrence of $z$, then pass an occurrence of $z'$ and finally reaches position $p$, where variable $z$ occurs. Regarding $D_{\alpha}$ this means that a factorisation $D_{\alpha} = \widetilde{\beta} \cdot z \cdot \widetilde{\gamma_1} \cdot z' \cdot \widetilde{\gamma_2} \cdot z \cdot \widetilde{\delta}$ is possible.\par
Next, we consider the case that it is not possible to factorise $\alpha = \widehat{\beta} \cdot z \cdot \widehat{\gamma_1} \cdot z' \cdot \widehat{\gamma_2} \cdot z \cdot \widehat{\delta}$. As explained above, this implies that $\alpha = \beta \cdot x \cdot \widehat{\gamma_1} \cdot z \cdot \widehat{\gamma_2} \cdot z' \cdot \widehat{\gamma_3} \cdot x \cdot \delta$ with $|\beta \cdot x \cdot \widehat{\gamma_1} \cdot z \cdot \widehat{\gamma_2}|_{z'} = 0$ and $|\widehat{\gamma_2} \cdot z' \cdot \widehat{\gamma_3} \cdot x \cdot \delta|_z = 0$. Let $r_z := |\beta \cdot x \cdot \widehat{\gamma_1}| + 1$ and $l_{z'} := |\beta \cdot x \cdot \widehat{\gamma_1} \cdot z \cdot \widehat{\gamma_2}| + 1$ be the positions of the variables $z$ and $z'$ pointed out in the factorisation above. Obviously, $r_z$ is the rightmost occurrence of $z$ and $l_{z'}$ is the leftmost occurrence of $z'$. These positions $r_z$ and $l_{z'}$ have to be covered by some matching positions in the complete matching order $(m_1, \ldots, m_k)$, i.\,e., there exist matching positions $m_i := (l_z, r_z)$ and $m_{i'} := (l_{z'}, r_{z'})$. We can assume that $r_z$ is the right element and $l_{z'}$ the left element of a matching position, as these positions describe the rightmost and the leftmost occurrences of the variable $z$ and $z'$, respectively. Moreover, $(m_1, \ldots, m_k)$ has to contain a complete matching order for variable $x$ in $\alpha$. Since there is no occurrence of $x$ in the factor $\gamma$, this implies the existence of a matching position $m_{i''} := (l_x, r_x)$ with $l_x \leq |\beta| + 1$ and $|\beta \cdot x \cdot \widehat{\gamma_1} \cdot z \cdot \widehat{\gamma_2} \cdot z' \cdot \widehat{\gamma_3}| + 1 \leq r_x$. We simply assume that $l_x = |\beta| + 1$ and $r_x = |\beta \cdot x \cdot \widehat{\gamma_1} \cdot z \cdot \widehat{\gamma_2} \cdot z' \cdot \widehat{\gamma_3}| + 1$, as this is no loss of generality regarding the following argumentation. Hence, we deal with the following situation (recall that $l_x$, $r_x$, $r_z$ and $l_{z'}$ are positions of $\alpha$):\\

\begin{tikzpicture}

\draw (-0.5, 0.4) node {$\alpha = $};
\draw[black] (0,0) -- (0,0.8);
\draw[black] (0,0.8) -- (10.4,0.8);
\draw[black] (10.4,0.8) -- (10.4,0);
\draw[black] (10.4,0) -- (0,0);
\draw[black] (1.5,0) -- (1.5,0.8);
\draw[black] (2.0,0) -- (2.0,0.8);
\draw[black] (3.8,0) -- (3.8,0.8);
\draw[black] (4.3,0) -- (4.3,0.8);
\draw[black] (6.1,0) -- (6.1,0.8);
\draw[black] (6.6,0) -- (6.6,0.8);
\draw[black] (8.4,0) -- (8.4,0.8);
\draw[black] (8.9,0) -- (8.9,0.8);
\draw (0.75, 0.4) node {$\beta$};
\draw (1.75, 0.4) node {$x$};
\draw (2.9, 0.4) node {$\widehat{\gamma_1}$};
\draw (4.05, 0.4) node {$z$};
\draw (5.2, 0.4) node {$\widehat{\gamma_2}$};
\draw (6.35, 0.4) node {$z'$};
\draw (7.5, 0.4) node {$\widehat{\gamma_3}$};
\draw (8.65, 0.4) node {$x$};
\draw (9.65, 0.4) node {$\delta$};

\draw[black] (1.75, -0.2) -- (1.75, -0.5);
\draw[black] (1.75, -0.2) -- (1.85, -0.3);
\draw[black] (1.75, -0.2) -- (1.65, -0.3);
\draw (1.75, -0.8) node {$l_{x}$};

\draw[black] ($(1.75, -0.2) + (2.3,0)$) -- ($(1.75, -0.5) + (2.3,0)$);
\draw[black] ($(1.75, -0.2) + (2.3,0)$) -- ($(1.85, -0.3) + (2.3,0)$);
\draw[black] ($(1.75, -0.2) + (2.3,0)$) -- ($(1.65, -0.3) + (2.3,0)$);
\draw ($(1.75, -0.8) + (2.3,0)$) node {$r_{z}$};

\draw[black] ($(1.75, -0.2) + (4.6,0)$) -- ($(1.75, -0.5) + (4.6,0)$);
\draw[black] ($(1.75, -0.2) + (4.6,0)$) -- ($(1.85, -0.3) + (4.6,0)$);
\draw[black] ($(1.75, -0.2) + (4.6,0)$) -- ($(1.65, -0.3) + (4.6,0)$);
\draw ($(1.75, -0.8) + (4.6,0)$) node {$l_{z'}$};

\draw[black] ($(1.75, -0.2) + (6.9,0)$) -- ($(1.75, -0.5) + (6.9,0)$);
\draw[black] ($(1.75, -0.2) + (6.9,0)$) -- ($(1.85, -0.3) + (6.9,0)$);
\draw[black] ($(1.75, -0.2) + (6.9,0)$) -- ($(1.65, -0.3) + (6.9,0)$);
\draw ($(1.75, -0.8) + (6.9,0)$) node {$r_{x}$};

\end{tikzpicture}

Now, in the same way as before, we interpret the Janus operating mode as a sequence of input head movements. We proceed by considering two cases concerning the order of the matching positions $m_{i'} = (l_{z'}, r_{z'})$ and $m_{i''} = (l_x, r_x)$ in the complete matching order, i.\,e., either $i' < i''$ or $i'' < i'$. In the latter case, $i'' < i'$, the right input head is moved from the leftmost variable in $\alpha$ to position $r_x$, hence, it passes $z$ and $z'$ in this order. Furthermore, the left input head is moved to position $l_x$. After that, since $i'' < i'$, the left input head has to be moved from position $l_x$ to position $l_{z'}$, thus, passing position $r_z$ where variable $z$ occurs. Hence, we conclude $D_{\alpha} = \widetilde{\beta} \cdot z \cdot \widetilde{\gamma_1} \cdot z' \cdot \widetilde{\gamma_2} \cdot z \cdot \widetilde{\delta}$. Next, we assume $i' < i''$, so the left input head is moved from the leftmost variable in $\alpha$ to position $l_{z'}$, so again, an input head passes $z$ and $z'$ in this order. After that, the left input head is moved from position $l_{z'}$ to position $l_x$, thus, it passes variable $z$ on position $r_z$. Again, we can conclude $D_{\alpha} = \widetilde{\beta} \cdot z \cdot \widetilde{\gamma_1} \cdot z' \cdot \widetilde{\gamma_2} \cdot z \cdot \widetilde{\delta}$.
\hfill \emph{q.e.d. (Claim)}\par\bigskip\noindent
Hence, for all $z, z' \in \Gamma$, $z \neq z'$, $D_{\alpha}$ can be factorised into $D_{\alpha} = \widetilde{\beta} \cdot z \cdot \widetilde{\gamma_1} \cdot z' \cdot \widetilde{\gamma_2} \cdot z \cdot \widetilde{\delta}$ or $D_{\alpha} = \widetilde{\beta} \cdot z' \cdot \widetilde{\gamma_1} \cdot z \cdot \widetilde{\gamma_2} \cdot z' \cdot \widetilde{\delta}$, and therefore we can apply Lemma~\ref{crossingLemma} and conclude that there exists a $z \in \Gamma$ such that $D_\alpha$ can be factorised into $D_\alpha = \overline{\beta} \cdot z \cdot \overline{\gamma}$ with $(\Gamma \slash \{z\}) \subseteq (\var(\overline{\beta}) \cap \var(\overline{\gamma}))$. This directly implies that $\cn(\Delta_{\alpha}) \geq |\Gamma| - 1 = \vd(\alpha)$. 
\end{proof}

In the previous section, the task of finding an optimal Janus automaton for a pattern was shown to be equivalent to finding an optimal Janus operating mode for this pattern. Now, by the above result, a Janus operating mode $\Delta_{\alpha}$ for some pattern $\alpha$ is optimal if $\cn(\Delta_{\alpha}) = \vd(\alpha)$ is satisfied. Hence, our next goal is to find a Janus operating mode with that property. To this end, we shall first define a special complete matching order from which the optimal Janus operating mode is then derived.
\begin{definition}\label{canonicalMatchingOrderDefinition}
Let $\alpha := y_1 \cdot y_2 \cdot \ldots \cdot y_n$ be a terminal-free pattern with $p := |\var(\alpha)|$. For every $x_i \in \var(\alpha)$, let $\varpos{i}(\alpha) := \{j_{i, 1}, j_{i, 2}, \ldots, j_{i, n_i}\}$ with $n_i := |\alpha|_{x_i}$, $j_{i, l} < j_{i, l + 1}$, $1 \leq l \leq n_i - 1$. Let $(m_1, m_2, \ldots, m_k)$, $k = \sum_{i = 1}^{p} n_i - 1$, be an enumeration of the set $\{(j_{i, l}, j_{i, l + 1})~|~1 \leq i \leq p, 1 \leq l \leq n_i - 1\}$ such that, for every $i'$, $1 \leq i' < k$, the left element of the pair $m_{i'}$ is smaller than the left element of $m_{i' + 1}$.
We call $(m_1, m_2, \ldots, m_k)$ the \emph{canonical matching order for $\alpha$}.
\end{definition}

\begin{proposition}\label{canonicalMOProp}
Let $\alpha$ be a terminal-free pattern. The canonical matching order for $\alpha$ is a complete matching order.
\end{proposition}

\begin{proof}
For every $x_i \in \var(\alpha)$, let $\varpos{i}(\alpha) := \{j_{i, 1}, j_{i, 2}, \ldots, j_{i, n_i}\}$ with $n_i := |\alpha|_{x_i}$, $j_{i, l} < j_{i, l + 1}$, $1 \leq l \leq n_i - 1$. The tuple
\begin{equation*}
((j_{i, 1}, j_{i, 2}), (j_{i, 2}, j_{i, 3}), \ldots, (j_{i, n_i - 2}, j_{i, n_i - 1}), (j_{i, n_i - 1}, j_{i, n_i}))
\end{equation*}
is clearly a matching order for $x_i$ in $\alpha$. As the canonical matching order contains all these matching orders for each variable $x_i \in \var(\alpha)$, it is a complete matching order for $\alpha$. 
\end{proof}

Intuitively, the canonical matching order can be constructed by simply moving through the pattern from left to right and for each encountered occurrence of a variable $x$, this occurrence and the next occurrence of $x$ (if there is any) constitutes a matching position. For instance, the canonical matching order for the example pattern $\beta$ introduced in Section~\ref{sec:pattern} is $((1,3), (2,4), (4,6), (5,7))$.\par 
We proceed with the definition of a Janus operating mode that is derived from the canonical matching order. Before we do so, we informally explain how this is done. To this end, we employ the interpretation of Janus operating modes as instructions for input head movements. In each step of moving the input heads from one matching position to another, we want to move first the left head completely and then the right head. This is not a problem as long as the part the left head has to be moved over and the part the right head has to be moved over are not overlapping. However, if they are overlapping, then the left head would overtake the right head which conflicts with the definition of Janus operating modes. So in this special case, we first move the left head until is reaches the right head and then we move both heads simultaneously. As soon as the left head reaches the left element of the next matching position, we can keep on moving the right head until it reaches the right element of the next matching position.

\begin{definition}\label{canonicalJanusOperatingMode}
Let $(m_1, m_2, \ldots, m_k)$ be the canonical matching order for a ter\-minal-free pattern $\alpha$.
For any $m_{i - 1} := (j'_1, j'_2)$ and $m_i := (j_1, j_2)$, $2 \leq i \leq k$, let $(p_1, p_2, \ldots, p_{k_1}) := g(j'_1, j_1)$ and $(p'_1, p'_2, \ldots, p'_{k_2}) := g(j'_2, j_2)$, where $g$ is the function introduced in Definition~\ref{janusOperatingModeDefinition}. If $j_1 \leq j'_2$, then we define
\begin{equation*}
D_i := ((p_1, \lambda), (p_2, \lambda), \ldots, (p_{k_1}, \lambda), (p'_1, \rho), (p'_2, \rho), \ldots, (p'_{k_2}, \rho), (j_2, \rho), (j_1, \lambda))\enspace.
\end{equation*}
If, on the other hand, $j'_2 < j_1$, we define $D_i$ in three parts
\begin{align*}
D_i := (&(p_1, \lambda), (p_2, \lambda), \ldots, (j'_2, \lambda), \\
&(j'_2 + 1, \rho), (j'_2 + 1, \lambda), (j'_2 + 2, \rho), (j'_2 + 2, \lambda), \ldots, (j_1 - 1, \rho), (j_1 - 1, \lambda), \\
&(j_1, \rho), (j_1 + 1, \rho), \ldots, (j_2 - 1, \rho), (j_2, \rho), (j_1, \lambda))\enspace.
\end{align*}
Finally, $D_1 := ((1, \rho), (2, \rho), \ldots, (j - 1, \rho), (j, \rho), (1, \lambda))$, where $m_1 = (1, j)$. The tuple $(D_1$, $D_2$, $\ldots$, $D_k)$ is called the \emph{canonical Janus operating mode}.
\end{definition}
If we derive a Janus operating mode from the canonical matching order $((1,3)$, $(2,4)$, $(4,6)$, $(5,7))$ for $\beta$ as described in Definition~\ref{canonicalJanusOperatingMode} we obtain the canonical Janus operating mode 
$(((1, \rho)$, $(2, \rho)$, $(3, \rho)$, $(1, \lambda))$, $((4, \rho)$, $(2, \lambda))$, $((3, \lambda)$, $(5, \rho)$, $(6, \rho)$, $(4, \lambda))$, $((7, \rho)$, $(5, \lambda)))$.
This canonical Janus operating mode has a counter number of $1$, so its counter number is smaller than the counter number of the example Janus operating mode $\Delta_{\beta}$ given in Section~\ref{sec:pattern} and, furthermore, equals the variable distance of $\beta$. 
Referring to Theorem~\ref{cnGeqVd}, we conclude that the canonical Janus operating mode for $\beta$ is optimal. 
The next lemma shows that this holds for every pattern. 

\begin{lemma}\label{canonicalCnEqualsVd}
Let $\alpha$ be a terminal-free pattern and let $\Delta_{\alpha}$ be the canonical Janus operating mode for $\alpha$. Then $\cn(\Delta_{\alpha}) = \vd(\alpha)$.
\end{lemma}

\begin{proof}
Let $\alpha := y_1 \cdot y_2 \cdot \ldots \cdot y_n$ and let $\overline{\Delta_{\alpha}} := ((d'_1, \mu'_1), (d'_2, \mu'_2), \ldots, (d'_{k'}, \mu'_{k'}))$ be the head movement indicator of the canonical Janus operating mode. This sequence $\overline{\Delta_{\alpha}}$ contains numbers $d'_i$, $1 \leq i \leq k'$, that are positions of $\alpha$, i.\,e.\ $1 \leq d'_i \leq |\alpha|$, $1 \leq i \leq k'$. Hence, we can associate a sequence of variables $(y_{d'_1}, y_{d'_2}, \ldots, y_{d'_{k'}})$ with $\overline{\Delta_{\alpha}}$.\par
In order to prove Lemma~\ref{canonicalCnEqualsVd}, we assume to the contrary that $\cn(\Delta_{\alpha}) > \vd(\alpha)$. This implies that there is a $p$, $1 \leq p \leq k'$, and a set $\Gamma$ of at least $\pi := \vd(\alpha) + 1$ different variables $z_1, z_2, \ldots, z_{\pi}$ such that $y_{d'_p} \notin \Gamma$ and, for every $z \in \Gamma$, there exist $j, j'$, $1 \leq j < p < j' \leq k'$, with $y_{d'_j} = y_{d'_{j'}} = z$.\par
We can interpret $\overline{\Delta_{\alpha}}$ as a sequence of input head movements over the pattern $\alpha$ as explained in Remark~\ref{inputHeadsPerspectiveRemark}. We are particularly interested in the position of the \emph{left} head in $\alpha$ at step $p$ of $\overline{\Delta_{\alpha}}$. Thus, we define $\widehat{p}$ such that $d'_{\widehat{p}} = \max\{d'_j \mid 1 \leq j \leq p, \mu'_j = \lambda\}$. However, we note that $\{d'_j \mid 1 \leq j \leq p, \mu'_j = \lambda\} = \emptyset$ is possible and in this case $d'_{\widehat{p}}$ would be undefined. So for now, we assume that $\{d'_j \mid 1 \leq j \leq p, \mu'_j = \lambda\} \neq \emptyset$ and consider the other case at the end of this proof. Moreover, we need to define the rightmost position in $\alpha$ that has been visited by any input head when we reach step $p$ in $\overline{\Delta_{\alpha}}$. By definition of the canonical matching order, this has to be the right input head, as it is always positioned to the right of the left input head. Thus, we define $p_{\max}$ such that $d'_{p_{\max}} := \max\{d'_j \mid 1 \leq j \leq p\}$.\par
Now, we can consider $\alpha$ in the factorisation
\begin{equation*}
\alpha = \beta \cdot y_{d'_{\widehat{p}}} \cdot \gamma \cdot y_{d'_{p_{\max}}} \cdot \delta\enspace.
\end{equation*}
By definition of the positions $\widehat{p}$ and $p_{\max}$ above, we can conclude the following. After performing all steps $d'_{j}$ with $1 \leq j \leq p$, position $d'_{\widehat{p}}$ is the position where the left head is located right now. This implies, by definition of the canonical Janus operating mode, that no head will be moved to one of the positions in $\beta$ again. The position $d'_{p_{\max}}$ is the rightmost position visited by any head so far. Hence, until now, no head has reached a position in $\delta$. \par
Regarding the sequence of variables $(y_{d'_1}, y_{d'_2}, \ldots, y_{d'_{k'}})$ we can observe that for every $j$, $1 \leq j \leq p$, $y_{d'_j} \in \var(\beta \cdot y_{d'_{\widehat{p}}} \cdot \gamma \cdot y_{d'_{p_{\max}}})$, and, for every $j'$, $p < j' \leq k'$, $y_{d'_{j'}} \in \var(\gamma \cdot y_{d'_{p_{\max}}} \cdot \delta)$. This follows directly from our interpretation of $\overline{\Delta_{\alpha}}$ as a sequence of input head movements over $\alpha$. 
Moreover, since for every $z \in \Gamma$, there exist $j, j'$, $1 \leq j < p < j' \leq k'$, with $y_{d'_j} = y_{d'_{j'}} = z$, we can conclude that $\Gamma \subseteq (\var(\beta \cdot y_{d'_{\widehat{p}}} \cdot \gamma \cdot y_{d'_{p_{\max}}}) \cap \var(\gamma \cdot y_{d'_{p_{\max}}} \cdot \delta))$. We can further show that $\Gamma \subseteq \var(\gamma \cdot y_{d'_{p_{\max}}})$. To this end, we assume that for some $z \in \Gamma$, $z \notin \var(\gamma \cdot y_{d'_{p_{\max}}})$, which implies $z \in (\var(\beta \cdot y_{d_{\widehat{p}}}) \cap \var(\delta))$. Hence, we can conclude that there exists a matching position $(l_z, r_z)$ in the canonical matching order, where the left element $l_z$ is a position in $\beta \cdot y_{d_{\widehat{p}}}$ and the right element $r_z$ is a position in $\delta$, i.\,e., $1 \leq l_z \leq |\beta \cdot y_{d_{\widehat{p}}}|$ and $|\beta \cdot y_{d'_{\widehat{p}}} \cdot \gamma \cdot y_{d'_{p_{\max}}}| + 1 \leq r_z \leq |\alpha|$. By definition of the canonical Janus operating mode, this implies that the rightmost position in $\alpha$, that has been visited by any input head when we reached step $p$ in $\overline{\Delta_{\alpha}}$ has to be at least position $r_z$. Since $r_z > d'_{p_{\max}}$, this is clearly a contradiction. Consequently, we conclude that $\Gamma \subseteq \var(\gamma \cdot y_{d'_{p_{\max}}})$.\par
We recall that position $d'_{p_{\max}}$ of $\alpha$ has already been reached by the right head and that in the canonical Janus operating mode, the right head is exclusively moved from the right element of some matching position $(l, r)$ to the right element of another matching position $(l', r')$. Consequently, either $r \leq d'_{p_{\max}} \leq r'$ or $r' \leq d'_{p_{\max}} \leq r$ and, furthermore, the left elements $l$ and $l'$ must be positions in the factor $\beta \cdot y_{d'_{\widehat{p}}}$. Thus, there has to be a matching position $(l, r)$ in the canonical matching order with $l \leq d'_{\widehat{p}}$ and $r \geq d'_{p_{\max}}$. Therefore, we can refine the factorisation from above by factorising $\beta \cdot y_{d'_{\widehat{p}}}$ into $\beta_1 \cdot y_l \cdot \beta_2$ and $y_{d'_{p_{\max}}} \cdot \delta$ into $\delta_1 \cdot y_r \cdot \delta_2$; thus, we obtain
\begin{equation*}
\alpha = \beta_1 \cdot y_l \cdot \beta_2 \cdot \gamma \cdot \delta_1 \cdot y_r \cdot \delta_2\enspace.
\end{equation*}
In the following, we show that the factor between the left and right element of the matching position $(l, r)$, i.\,e., $\beta_2 \cdot \gamma \cdot \delta_1$, contains too many distinct variables different from $y_l = y_{r}$. More precisely, the number of such variables is clearly bounded by the variable distance, but, by means of the variables in $\Gamma$, we obtain a contradiction by showing that there are $\vd(\alpha) + 1$ such variables in the factor $\beta_2 \cdot \gamma \cdot \delta_1$. To this end, we first recall that we have already established that $\Gamma \subseteq \var(\gamma \cdot y_{d'_{p_{\max}}})$ and, furthermore,  $y_{d'_p} \notin \Gamma$ and $(l, r)$ is a matching position; thus, $y_l = y_r$. \par
By the factorisation above, we know that $d'_{p_{\max}} \leq r$. If $d'_{p_{\max}} < r$, then $\Gamma \subseteq \var(\gamma \cdot y_{d'_{p_{\max}}})$ implies $\Gamma \subseteq \var(\gamma \cdot \delta_1)$. We can further note, that $y_r$ cannot be an element of $\Gamma$ as this contradicts to the fact that $(l, r)$ is a matching position. Thus, we have $|\Gamma|$ variables different from $y_l = y_r$ occurring in $\beta_2 \cdot \gamma \cdot \delta_1$ and we obtain the contradiction as described above. \par
In the following, we assume that $d'_{p_{\max}} = r$ and note that this implies $\delta_1 = \varepsilon$. We observe that there are two cases depending on whether or not $y_{d'_{p_{\max}}} \in \Gamma$. We start with the easy case, namely $y_{d'_{p_{\max}}} \notin \Gamma$, and note that in this case $\Gamma \subseteq \var(\gamma \cdot y_{d'_{p_{\max}}})$ implies $\Gamma \subseteq \var(\gamma)$. In the same way as before, this leads to 
a contradiction.\par
It remains to consider the case that $y_{d'_{p_{\max}}} \in \Gamma$. Here, $\Gamma \subseteq \var(\gamma)$ is not satisfied anymore, as $(l, d'_{p_{\max}})$ is a matching position (recall that we still assume  $d'_{p_{\max}} = r$) and, thus, $y_{d'_{p_{\max}}} \notin \var(\gamma)$. In the following we consider the variable $y_{d'_p}$, for which, by definition, $y_{d'_p} \notin \Gamma$ is satisfied. Hence, in order to obtain a contradiction, it is sufficient to show that $y_{d'_p} \in \var(\beta_2 \cdot \gamma \cdot \delta_1)$. To this end, we need the following claim:\par\bigskip\noindent
\emph{Claim} $l \leq d'_{p}$. \par\medskip\noindent
\emph{Proof (Claim).}  If $\mu'_{p} = \lambda$, then, by definition, $d'_{\widehat{p}} = d'_{p}$ and if $\mu'_{p} = \rho$, then $d'_{\widehat{p}} < d'_{p}$, since $\widehat{p}$ is the position of the left head and $d'_{p}$ is the position of the right head. Hence, since $l \leq d'_{\widehat{p}}$, we conclude $l \leq d'_{\widehat{p}} \leq d'_{p}$. \hfill \emph{q.e.d. (Claim)} \par\bigskip
If $l < d'_{p}$, then $y_{d'_p} \in \var(\beta_2 \cdot \gamma \cdot \delta_1)$, since $y_{d'_p} = y_{d'_{p_{\max}}}$ is not possible as, by assumption, $y_{d'_{p_{\max}}} \in \Gamma$ and $y_{d'_p} \notin \Gamma$. Hence, we assume $l = d'_{p}$, which implies $y_l = y_{d'_{p}}$. We can show that this is a contradiction. First, we recall that $(l, d'_{p_{\max}})$ is a matching position, so $y_l = y_{d'_{p_{\max}}}$ and since $y_{d'_{p_{\max}}} \in \Gamma$, $y_{l} \in \Gamma$ as well. Furthermore, $y_{d'_{p}} \notin \Gamma$, which contradicts $y_l = y_{d'_{p}}$. We conclude that $y_{d'_p} \in \var(\beta_2 \cdot \gamma \cdot \delta_1)$ must be satisfied.\par
Hence, for each possible case, we obtain $|\var(\beta_2 \cdot \gamma \cdot \delta_1)| \geq \pi$, which is a contradiction.\par
It still remains to consider the case $\{d'_j \mid 1 \leq j \leq p, \mu'_j = \lambda\} = \emptyset$. In this case we have $\mu'_i = \rho$ for every $i$ with $1 \leq i \leq p$. This implies that until now the left input head has not yet entered $\alpha$ and the right head has been moved directly from the first position of $\alpha$ to position $d'_p$ without reversing direction. Furthermore, we know that the first matching position of the canonical matching order is $(1, r)$, where $d'_p \leq r$.\par
If $d'_p = r$, we can factorise $\alpha$ into 
\begin{equation*}
\alpha = y_1 \cdot \beta \cdot y_{d'_p} \cdot \gamma\,,
\end{equation*}
where $(1, d'_p)$ is a matching position. As for every $z \in \Gamma$ there exists an $i$, $1 \leq i < p$, with $y_{d_i} = z$ and since $y_{d'_p} \notin \Gamma$, we conclude $\Gamma \subseteq \var(\beta)$. This directly implies $\vd(\alpha) \geq \pi$, which is a contradiction.\par
If, on the other hand, $d'_p < r$, then we can factorise $\alpha$ into
\begin{equation*}
\alpha = y_1 \cdot \beta_1 \cdot y_{d'_p} \cdot \beta_2 \cdot y_r \cdot \gamma\,.
\end{equation*}
In the same way as before, we can conclude that $\Gamma \subseteq \var(y_1 \cdot \beta_1)$, thus, $(\Gamma \slash \{y_1\}) \subseteq \var(\beta_1)$. Now, as $y_{d'_p} \notin \Gamma$, we have $(\Gamma \slash \{y_1\}) \cup \{y_{d'_p}\} \subseteq \var(\beta_1 \cdot y_{d'_p} \cdot \beta_2)$, where $|(\Gamma \slash \{y_1\}) \cup \{y_{d'_p}\}| = \pi$ and, since $(1,r)$ is a matching position, $\vd(\alpha) \geq \pi$ follows, which is a contradiction. This concludes the proof of Lemma~\ref{canonicalCnEqualsVd}.
\end{proof}

The above lemma, in conjunction with Theorems~\ref{janusConstructionTheorem}~and~\ref{cnGeqVd}, shows that the canonical Janus operating mode for a pattern $\alpha$ can be transformed into a Janus automaton that is optimal with respect to the number of counters. We subsume this first main result in the following theorem:

\begin{theorem}\label{vdCorollary}
Let $\alpha$ be a terminal-free pattern. There exists a $\jfa(\vd(\alpha) + 1)$ $M$ such that $L(M) = L(\alpha)$.
\end{theorem}

The Janus automaton obtained from the canonical Janus operating mode for a pattern $\alpha$ (in the way it is done in the proof of Theorem~\ref{janusConstructionTheorem}) is called the \emph{canonical Janus automaton}. As already stated above, Theorem~\ref{vdCorollary} shows the optimality of the canonical automaton. However, this optimality is subject to a vital assumption: we assume that the automaton needs to know the length of a factor in order to move an input head over this factor. Although this assumption is quite natural, we shall reconsider it in more detail in Section~\ref{sec:conclusion}.\par
As stated in Section~\ref{sec:intro}, the variable distance is the crucial parameter when constructing canonical Janus automata for pattern languages. We obtain a polynomial time match test for any class of patterns with a restricted variable distance:
\begin{theorem}\label{mainComplexityResult}
There is a computable function that, given any terminal-free pattern $\alpha$ and $w \in \Sigma^*$, decides on whether $w \in L(\alpha)$ in time $\landau(|\alpha|^3\,|w|^{(\vd(\alpha) + 4)})$.
\end{theorem}

\begin{proof}
We present an algorithm solving the membership problem for terminal-free pattern languages within the time bound claimed in Theorem~\ref{mainComplexityResult}. Our algorithm, on input $\alpha$ and $w$, simply constructs the canonical Janus automaton $M$ for $\alpha$ and then solves the acceptance problem for $M$ on input $w$. As $L(M) = L(\alpha)$, this algorithm clearly works correctly.\par
Regarding the time complexity we have to investigate two aspects: Firstly, the time complexity of transforming $\alpha$ into the canonical Janus automaton $M$ and, secondly, the time complexity of solving the acceptance problem for $M$ on input $w$. To simplify the estimations of time complexities, we define $n := |w|$. In the strict sense, the input has length $|w| + 2$ and there are $|w| + 1$ possible counter bounds to guess, but as we shall use the Landau notation, $n$ is sufficiently accurate for the following analysis.\par
We begin with transforming $\alpha := y_1 \cdot y_2 \cdot \ldots \cdot y_{n'}$ into $M$. To this end, we construct the canonical matching order $(m_1, m_2, \ldots, m_k)$, which can be obtained from $\alpha$ in time $\landau(|\alpha|)$. Definition~\ref{canonicalJanusOperatingMode} shows that the canonical Janus operating mode $\Delta_{\alpha} := (D_1, \ldots, D_k)$ can be directly constructed from the canonical matching order and the time complexity required to do so is merely the size of $\Delta_{\alpha}$. Obviously, every $D_i$, $1 \leq i \leq k$, has $\landau(|\alpha|)$ elements and $k \leq |\alpha|$. Thus, we conclude that $\Delta_{\alpha}$ can be constructed in $\landau(|\alpha|^2)$. Let $\overline{\Delta_{\alpha}} = ((d'_1, \mu'_1), (d'_2, \mu'_2), \ldots, (d'_{k'}, \mu'_{k'}))$ be the head movement indicator of $\Delta_{\alpha}$, and let $D_{\alpha} := y_{d'_1} \cdot y_{d'_2} \cdot \ldots \cdot y_{d'_{k'}}$, where, as described above, $k' \leq |\alpha|^2$. Next, we have to construct a mapping $\co : \var(\alpha) \rightarrow  \{1, \ldots, \vd(\alpha) + 1\}$ with the required properties described in the proof of Theorem~\ref{janusConstructionTheorem}, i.\,e., if, for some $z, z' \in \var(\alpha)$, $z \neq z'$, $D_{\alpha}$ can be factorised into $D_{\alpha} = \beta \cdot z \cdot \gamma \cdot z' \cdot \gamma' \cdot z \cdot \delta$, then $\co(z) \neq \co(z')$. Such a mapping can be constructed in the following way. Assume that it is possible to mark counters either as free or as occupied. We move over the pattern $y_{d_1} \cdot y_{d_2} \cdot \ldots \cdot y_{d_{k'}}$ from left to right and whenever a variable $x_i$ is encountered for the first time, we set $\co(x_i) := j$ for some counter $j$ that is not occupied right now and then mark this counter $j$ as occupied. Whenever a variable $x_i$ is encountered for the last time, counter $\co(x_i)$ is marked as free. As we have to move over $\overline{\Delta_{\alpha}}$ in order to construct $\co$ in this way, time $\landau(k') = \landau(|\alpha|^2)$ is sufficient. We note that this method can be applied as it is not possible that there are more than $\cn(\Delta_{\alpha}) + 1 = \vd(\alpha) + 1$ variables such that for all $z, z'$, $z \neq z'$ of them, $D_{\alpha}$ can be factorised into $D_{\alpha} = \beta \cdot z \cdot \gamma \cdot z' \cdot \gamma' \cdot z \cdot \delta$ or $D_{\alpha} = \beta \cdot z' \cdot \gamma \cdot z \cdot \gamma' \cdot z' \cdot \delta$. This can be shown in the same way as we have already done in the proof of Theorem~\ref{janusConstructionTheorem}.\par
Next we transform each $D_p$, $1 \leq p \leq k$, into a part of the automaton $M$, following the construction in the proof of Theorem~\ref{janusConstructionTheorem}. For the remainder of this proof, we define $\pi := \vd(\alpha) + 1$. We show how many states are needed to implement an arbitrary $D_p$ with $p \geq 2$. Therefore, we define 
\begin{equation*}
D_p := ((j_1, \mu_1), (j_2, \mu_2), \ldots, (j_{k''}, \mu_{k''}), (j_r, \rho), (j_l, \lambda))
\end{equation*}
with $\mu_q \in \{\lambda, \rho\}$, $1 \leq q \leq k''$, and the tuples $(j'_r, \rho)$, $(j'_l, \lambda)$ to be the last two elements of $D_{p - 1}$. We need the following sets of states.
\begin{align*}
&Q_{p, l} := \begin{cases}
\{\lforth_{p,q} \mid 1 \leq q \leq k'', \mu_q = \lambda \}& \mbox{if $j'_l < j_l$}\enspace,\\
\{\lback_{p,q} \mid 1 \leq q \leq k'', \mu_q = \lambda \}& \mbox{else}\enspace.
\end{cases}\\
&Q_{p, r} := \begin{cases}
\{\rforth_{p,q} \mid 1 \leq q \leq k'', \mu_q = \rho \}& \mbox{if $j'_r < j_r$}\enspace,\\
\{\rback_{p,q} \mid 1 \leq q \leq k'', \mu_q = \rho \}& \mbox{else}\enspace.
\end{cases}\\
&Q_p := Q_{p, l} \cup Q_{p, r} \cup \{\match_p\}\enspace.
\end{align*}
The set $Q_1$ is defined analogously, with the only difference that only forth-states are needed. Clearly, $|Q_p| = k'' + 1 = \landau(|\alpha|)$, $1 \leq p \leq k$. So as $k = \sum_{i = 1}^{|\var(\alpha)|} (|\alpha|_{x_i} - 1) = |\alpha| - |\var(\alpha)| \leq |\alpha|$, we can conclude that $|Q| = \landau(|\alpha|^2)$, where $Q := \bigcup_{i = 1}^{k} Q_i$. For each element $y$ in $(|Q| \times \{0, 1, \ldots, n + 1\}^2 \times \{\mathtt{t_{=}}, \mathtt{t_{<}}\}^{\pi})$ we need to define $\delta(y)$, so $\delta$ can be constructed in time $\landau(|\alpha|^2\,n^2\,2^\pi)$.
This shows that the automaton $M$ can be constructed in time $\landau(|\alpha|^2\,n^2\,2^\pi)$.\par
Next we shall investigate the time complexity of solving the acceptance problem for $M$ on input $w$. We apply the following idea. We construct a directed graph of possible configurations of $M$ as vertices, connected by an edge if and only if it is possible to get from one configuration to the other by applying the transition function $\delta$. Then we search this graph for a path leading from the initial configuration to a final configuration, i.\,e., an accepting path. For an arbitrary vertex $v$, we denote the number of edges starting at $v$ by \emph{outdegree of $v$} and the number of edges ending at $v$ by \emph{indegree of $v$}. The nondeterminism of the computation of $M$ is represented by the fact that there are vertices with outdegree greater than $1$, namely those configurations where a new counter bound is guessed. So the existence of an accepting path is a sufficient and necessary criterion for the acceptance of the input word $w$. Searching this graph for an accepting path leads to a \emph{deterministic} algorithm correctly solving the acceptance problem for $M$. Let $(V, E)$ be this graph. The problem of finding an accepting path can then be solved in time $\landau(|V| + |E|)$. We illustrate this idea more formally and define the set of vertices, i.\,e., the set of all possible configurations of $M$ on input $w$:
\begin{align*}
\widehat{C}'_{M, w} := \{(q, h_1, h_2, (c_1, C_1), \ldots, (c_{\pi}, C_{\pi}))~|~&q \in Q, 0 \leq h_1 \leq h_2 \leq n + 1,\\
&0 \leq c_i \leq C_i \leq n, 1 \leq i \leq \pi\}\enspace.
\end{align*}
Now we obtain $\widehat{C}_{M, w}$ by simply deleting all the configurations of $\widehat{C}'_{M, w}$ that cannot be reached in any computation of $M$ on input $w$. How this can be done shall be explained at the end of the proof. Furthermore, we define a set of edges $\widehat{E}_{M, w}$, connecting the configurations in $\widehat{C}_{M, w}$ as follows: for all $\widehat{c}_1, \widehat{c}_2 \in \widehat{C}_{M, w}$, $(\widehat{c}_1, \widehat{c}_2) \in \widehat{E}_{M, w}$ if and only if $\widehat{c}_1 \vdash_{M, w} \widehat{c}_2$. We call $\widehat{G}_{M, w} := (\widehat{C}_{M, w}, \widehat{E}_{M, w})$ the \emph{full computation graph of $M$ on input $w$}. To analyse the time complexity of searching $\widehat{G}_{M, w}$ for an accepting path, we have to determine the size of $\widehat{C}_{M, w}$ and $\widehat{E}_{M, w}$. By the construction given in the proof of Theorem~\ref{janusConstructionTheorem}, for all configurations $(q, h_1, h_2, (c_1, C_1), \ldots, (c_{\pi}, C_{\pi})) \in \widehat{C}_{M, w}$, there is at most one $i$, $1 \leq i \leq \pi$, with $c_i \geq 1$. That is due to the fact that when $M$ increments a counter, then this counter is incremented until the counter value jumps back to $0$ again before another counter is incremented. Thus, for each $i$, $1 \leq i \leq \pi$, there are $|Q|\,n^{\pi + 3}$, possible configurations $(q, h_1, h_2, (c_1, C_1), \ldots, (c_{\pi}, C_{\pi}))$ such that $c_i \geq 1$. Therefore, we obtain
\begin{equation*}
|\widehat{C}_{M, w}| = \landau(|Q|\,\pi\,n^{\pi + 3}) = \landau(|\alpha|^2\,(\vd(\alpha) + 1)\,n^{\pi + 3}) = \landau(|\alpha|^3\,n^{\pi + 3})\enspace.
\end{equation*}
Next, we analyse the number of edges in $\widehat{G}_{M, w}$. As already mentioned, due to the nondeterminism of Janus automata, there are vertices in $\widehat{G}_{M, w}$ with an outdegree greater than one. One such vertex is the initial configuration, as in the initial configuration, all $\pi$ counters are reset. Thus, the initial configuration has outdegree of $\landau(n^{\pi})$. Furthermore, if $M$ resets a counter by changing from one configuration $\widehat{c_1}$ to another configuration $\widehat{c_2}$, then $\widehat{c_1}$ has outdegree greater than one.
However, there is at most one counter reset by changing from one configuration to another, so, for these configurations, the outdegree is bounded by $n$. We know that $M$ has $|\var(\alpha)|$ states such that a counter is reset in this state and, furthermore, if a counter is reset, all counter values are $0$. Hence the number of configurations with outdegree $n$ is $\landau(|\var(\alpha)|\,n^{\pi + 2})$ and so we count $\landau(|\var(\alpha)|\,n^{\pi + 3})$ edges for these configurations. Finally, all the other vertices not considered so far have outdegree $1$, and, as the complete number of vertices is $\landau(|\alpha|^3\,n^{\pi + 3})$, we can conclude that the number of vertices with outdegree $1$ does not exceed $\landau(|\alpha|^3\,n^{\pi + 3})$. We obtain
\begin{equation*}
|\widehat{E}_{M, w}| = \landau(n^{\pi} + |\var(\alpha)|\,n^{\pi + 3} + |\alpha|^3\,n^{\pi + 3}) = \landau(|\alpha|^3\,n^{\pi + 3})\enspace.
\end{equation*}
Consequently, $\landau(|\widehat{C}_{M, w}| + |\widehat{E}_{M, w}|) = \landau(|\alpha|^3\,n^{\pi + 3})$ and, as $\pi = \vd(\alpha) + 1$, $\landau(|\widehat{C}_{M, w}| + |\widehat{E}_{M, w}|) = \landau(|\alpha|^3\,n^{\vd(\alpha) + 4})$. However, it remains to explain how exactly we can search the graph for an accepting path. This can be done in the following way. We start with the initial configuration of $M$ on input $w$ and then we construct the graph $\widehat{G}_{M, w}$ step by step by using a Depth-First-Search approach. By this method an accepting configuration is found if there exists one and, furthermore, we do not need to construct the whole set of configurations $\widehat{C}'_{M, w}$ first. This concludes the proof.
\end{proof}

This main result also holds for more general classes of extended regular expressions, e.\,g., those containing terminal symbols (see our example in Section~\ref{sec:janus}) or imposing regular restrictions to the sets of words variables can be substituted with, i.\,e., for every variable $x \in \var(\alpha)$ a regular language $R_{x}$ is given and the pattern describes then the set of all words $w$ that can be obtained from $\alpha$ by substituting every $x \in \var(\alpha)$ by some word in $R_x$. We anticipate, though, that the necessary amendments to our definitions involve some technical hassle.

\section{Conclusions}\label{sec:conclusion}

In the present work, we have studied an important NP-complete problem, namely the match test for extended regular expressions. We have pointed out that the match test shows the same characteristics as the membership problem for terminal-free pattern languages, and therefore we have restricted our technical considerations to the latter problem, which can be defined in a more concise manner. We have introduced the concept of the variable distance of a pattern, and our studies have revealed that the complexity of the membership problem is essentially determined by this subtle combinatorial property. Any restriction of this parameter has yielded major classes of pattern languages (and, hence, of extended regular expressions) with a polynomial-time match test. \par
We have also been able to prove our approach to be optimal. However, this optimality is subject to the following vital assumption. We assumed that a Janus automaton needs to know the length of a factor in order to move an input head over this factor and, thus, needs to store this length in form of a counter bound. Although this assumption is quite natural, it might be worthwhile to consider possibilities to abandon it. For instance, a Janus automaton is able to detect the left and right end of its input by means of the endmarkers. Therefore, it can move an input head from any position to either end of the input without using any counter. So if an input head has to be moved from one position to another, there are three ways of doing this. We can either move it directly over the intermediate factors (how it is done in the original definition of Janus operating modes) or we can move it first to either the left or the right endmarker and then from there to the new position. In the latter two cases, only the information of the lengths of the factors between the left endmarker or the right endmarker and the target position are required. It is straightforward to extend the definition of Janus operating modes in accordance with these new ideas. Furthermore, we could again use the concept of the counter number of Janus operating modes and transform these refined Janus operating modes into Janus automata in a similar way as done in the proof of Theorem~\ref{janusConstructionTheorem}. The following example points out that, using this new approach, we can find Janus automata with less counters than the canonical Janus automata.

\begin{example}
Let $\alpha := x_1 \cdot x_2 \cdot x_3 \cdot x_1 \cdot x_2 \cdot x_4 \cdot x_4 \cdot x_5 \cdot x_5 \cdot x_3$. Clearly, $\vd(\alpha) = 4$, thus the canonical Janus automaton for $\alpha$ needs $5$ counters. We observe that there exists a $\jfa(4)$ $M$ with $L(M) = L(\alpha)$. This automaton $M$ matches factors according to the complete matching order $((1,4),(2,5),(6,7),(8,9),(3,10))$. The trick is that after matching the factors related to the matching position $(6,7)$, i.\,e., the factors corresponding to the occurrences of $x_4$, the counter responsible for factors corresponding to $x_4$ is reused to match the factors related to the matching position $(8,9)$. Hence, so far, we only needed $4$ counters, but, obviously, we lost the information of the length of factors corresponding to $x_4$. Now, we find the situation that it still remains to match the factors corresponding to the occurrences of $x_3$, i.\,e.\ the matching position $(3,10)$, but we cannot simply move the left head back to factor $3$, as the automaton does not know the length of the factors corresponding to $x_4$ anymore. However, we can move it to the left endmarker first, and then from there, over the factors corresponding to $x_1$ and $x_2$, to factor $3$. 
We can do this without storing the lengths of factors related to $x_4$ and $x_5$. Hence, $4$ counters are sufficient.
\end{example}
The above illustrated amendments to our approach further complicate the definition of Janus operating modes and we do not know anymore how to efficiently compute the Janus operating mode that is optimal with respect to the counter number. An exhaustive search of all Janus operating modes is inappropriate, as we would have to deal with a vast number of possible such Janus operating modes. In summary, we anticipate that these potential amendments to our approach lead to very challenging technical problems, and therefore we leave them for future research.

\bibliographystyle{elsarticle-num}
\bibliography{MSbib}

\begin{thebibliography}{10}
\expandafter\ifx\csname url\endcsname\relax
  \def\url#1{\texttt{#1}}\fi
\expandafter\ifx\csname urlprefix\endcsname\relax\def\urlprefix{URL }\fi
\expandafter\ifx\csname href\endcsname\relax
  \def\href#1#2{#2} \def\path#1{#1}\fi

\bibitem{rei:apo}
D.~Reidenbach, M.~L. Schmid, A polynomial time match test for large classes of
  extended regular expressions, in: Proc. 15th International Conference on
  Implementation and Application of Automata, CIAA 2010, Vol. 6482 of Lecture
  Notes in Computer Science, 2011, pp. 241--250.

\bibitem{cam:afo}
C.~C\^ampeanu, K.~Salomaa, S.~Yu, A formal study of practical regular
  expressions, International Journal of Foundations of Computer Science 14
  (2003) 1007--1018.

\bibitem{aho:alg}
A.~Aho, Algorithms for finding patterns in strings, in: J.~van Leeuwen (Ed.),
  Handbook of Theoretical Computer Science, Vol. A: Algorithms and Complexity,
  MIT Press, 1990, pp. 255--300.

\bibitem{fri:mas}
J.~E.~F. Friedl, Mastering Regular Expressions, 3rd Edition, O'Reilly,
  Sebastopol, CA, 2006.

\bibitem{goo:re2}
R.~Cox, RE2, Google, \url{http://code.google.com/p/re2/} (2010).

\bibitem{lem:reg}
V.~L. Maout, Regular expressions at their best: A case for rational design, in:
  Proc. 15th International Conference on Implementation and Application of
  Automata, CIAA 2010, Vol. 6482 of Lecture Notes in Computer Science, 2011,
  pp. 310--320.

\bibitem{fre:ext}
D.~D. Freydenberger, Extended regular expressions: Succinctness and
  decidability, in: 28th International Symposium on Theoretical Aspects of
  Computer Science, STACS 2011, Vol.~9 of LIPIcs, 2011, pp. 507--518.

\bibitem{ang:fin2}
D.~Angluin, Finding patterns common to a set of strings, Journal of Computer
  and System Sciences 21 (1980) 46--62.

\bibitem{jia:pat}
T.~Jiang, E.~Kinber, A.~Salomaa, K.~Salomaa, S.~Yu, Pattern languages with and
  without erasing, International Journal of Computer Mathematics 50 (1994)
  147--163.

\bibitem{iba:ano}
O.~Ibarra, T.-C. Pong, S.~Sohn, A note on parsing pattern languages, Pattern
  Recognition Letters 16 (1995) 179--182.

\bibitem{ehr:fin}
A.~Ehrenfeucht, G.~Rozenberg, Finding a homomorphism between two words is
  {NP}-complete, Information Processing Letters 9 (1979) 86--88.

\bibitem{shi:pol1}
T.~Shinohara, Polynomial time inference of extended regular pattern languages,
  in: Proc. RIMS Symposium on Software Science and Engineering, Vol. 147 of
  Lecture Notes in Computer Science, 1982, pp. 115--127.

\bibitem{shi:pol2}
T.~Shinohara, Polynomial time inference of pattern languages and its
  application, in: Proc. 7th IBM Symposium on Mathematical Foundations of
  Computer Science, 1982, pp. 191--209.

\bibitem{iba:ont}
O.~Ibarra, On two-way multihead automata, Journal of Computer and System
  Sciences 7 (1973) 28--36.

\bibitem{hop:int2}
J.~Hopcroft, R.~Motwani, J.~Ullman, Introduction to Automata Theory, Languages,
  and Computation, Addison-Wesley, 2000.

\end{thebibliography}
 
\end{document}